\documentclass{article}
\usepackage[utf8]{inputenc}
\usepackage{authblk}
\usepackage{enumerate}
\usepackage{amsmath}
\usepackage{amsfonts}
\usepackage{nicefrac}
\usepackage{amssymb}
\usepackage{amsthm}
\usepackage{tikz}
\usetikzlibrary{calc}
\usepackage{pgfplots}
\pgfplotsset{compat=newest}
\usepackage{pgfplotstable}
\usepackage{xcolor}
\usepackage{thm-restate}
\usepackage{tikz}
\usepackage[square,numbers]{natbib}
\usepackage{tabularx}
\usepackage{booktabs}
\usepackage{mathtools}
\usepackage{blkarray}
\usepackage{cleveref}

\usepackage{algorithm}
\usepackage{algorithmicx}
\usepackage{algpseudocode}

\newcommand{\lv}[1]{}

\usepackage[a4paper,top=2cm,bottom=3cm,left=3cm,right=3cm,marginparwidth=1.75cm]{geometry}

\usepackage[backgroundcolor=gray!10,textsize=footnotesize]{todonotes}

\newtheorem{theorem}{Theorem}
\newtheorem{lemma}{Lemma}
\newtheorem{proposition}{Proposition}

\crefname{rrule}{Rule}{Rules}

\DeclarePairedDelimiterX{\abs}[1]{\lvert}{\rvert}{#1}
\DeclarePairedDelimiterX{\norm}[1]{\lVert}{\rVert}{#1}
\DeclarePairedDelimiterX{\ceil}[1]{\lceil}{\rceil}{#1}
\DeclarePairedDelimiterX{\angled}[1]{\langle}{\rangle}{#1}

\newcommand{\poly}{\ensuremath{\textnormal{poly}}}

\newcommand{\ncc}{\nu} %

\DeclareMathOperator{\pf}{pf}

\DeclareMathOperator{\dist}{dist}

\DeclareMathOperator{\GF}{GF}

\DeclareMathOperator{\vc}{vc}

\DeclareMathOperator{\td}{td}

\title{Fully Polynomial-time Algorithms Parameterized by Vertex Integrity Using Fast Matrix Multiplication}
\author[1]{Matthias Bentert\footnote{Supported by the European Research Council (ERC) project LOPRE (819416) under the Horizon 2020 research and innovation program.}}
\author[2]{Klaus Heeger\footnote{Supported by Deutsche Forschungsgemeinschaft (DFG) project FPTinP (NI 369/16).}}
\author[2]{Tomohiro Koana\footnote{Supported by the DFG project DiPa (NI 369/21).}}
\affil[1]{University of Bergen}
\affil[2]{Technische Universit\"at Berlin}
\date{\small \texttt{matthias.bentert@uib.no, heeger@post.bgu.ac.il, tomohiro.koana@tu-berlin.de}}

\begin{document}

\maketitle

\begin{abstract}
We study the computational complexity of several polynomial-time-solvable graph problems parameterized by \emph{vertex integrity}, a measure of a graph's vulnerability to vertex removal in terms of connectivity.
Vertex integrity is the smallest number $\iota$ such that there is a set~$S$ of~$\iota' \le \iota$ vertices such that every connected component of~$G-S$ contains at most $\iota-\iota'$~vertices.
It is known that the vertex integrity lies between the well-studied parameters \emph{vertex cover number} and \emph{tree-depth}.

Alon and Yuster [ESA 2007] designed algorithms for graphs with small vertex cover number using fast matrix multiplications.
We demonstrate that fast matrix multiplication can also be effectively used when parameterizing by vertex integrity~$\iota$ by developing efficient algorithms for problems including an~$O(\iota^{\omega-1}n)$-time algorithm for computing the girth of a graph, randomized~$O(\iota^{\omega - 1}n)$-time algorithms for \textsc{Maximum Matching} and for finding any induced four-vertex subgraph except for a clique or an independent set, and an~$O(\iota^{(\omega-1)/2}n^2) \subseteq O(\iota^{0.687} n^2)$-time algorithm for \textsc{All-Pairs Shortest Paths}.
These algorithms can be faster than previous algorithms parameterized by tree-depth, for which fast matrix multiplication is not known to be effective. 
\end{abstract}

\section{Introduction}
Parameterized complexity provides a powerful framework for studying NP-hard problems. %
The main idea behind parameterized algorithms is to analyze the running time in terms of the input size $|I|$ as well as a parameter $k$, some measure of the input instance.
A problem is \emph{fixed-parameter tractable} or \emph{FPT} for short, if it admits an \emph{FPT algorithm}, i.e., an algorithm running in~$f(k) \cdot |I|^{O(1)}$~time, where $f$ is a computable function solely depending on $k$.
In the past decade, a line of research dubbed ``FPT in P'' has emerged, where the goal is a more refined parameterized analysis of polynomial-time-solvable problems~\cite{DBLP:conf/soda/AbboudWW16,DBLP:journals/jcss/BentertFNN19,DBLP:journals/algorithmica/BringmannHM20,DBLP:journals/talg/CoudertDP19,DBLP:journals/talg/FominLSPW18,DBLP:conf/icalp/HegerfeldK19,DBLP:conf/stacs/KratschN20}.
Although the function~$f$ usually has to be at least exponential when working with NP-hard problems, this is not true for problems in~$P$.
FPT algorithms where~$f$ is a polynomial function are called \emph{fully polynomial-time algorithms} and are of course desirable.

We study graph problems in this work.
Let $n$ and $m$ be the number of vertices and edges, respectively.
Also, let $\vc$ be the vertex cover number and $\td$ be the tree-depth (see \Cref{sec:prelim} for definitions).
Alon and Yuster \cite{DBLP:conf/esa/AlonY07} demonstrated that fast matrix multiplication can be used effectively for graphs with a (not necessarily small) vertex cover, developing algorithms for \textsc{Maximum Matching} and \textsc{All-Pairs Shortest Paths (APSP)} that run in $O(n^{\omega})$ time (where $\omega < 2.372$ is the matrix multiplication exponent) even when $\vc = \Theta(n)$ (and even faster if the vertex cover number is smaller).
More recently, Iwata et al.~\cite{DBLP:conf/stacs/IwataOO18} proposed a divide-and-conquer framework in the design of fully polynomial-time algorithms parameterized by tree-depth.
For instance, they showed that \textsc{Maximum Matching} can be solved in $O(m \cdot \td)$ time.

In this work, we consider the parameter  \emph{vertex integrity}, a parameter that lies between vertex cover number and tree-depth.
The vertex integrity~$\iota$ of a graph~$G$ is the smallest integer such that~$G$~contains a set~$S$ of~$\iota' \le \iota$ vertices whose deletion results in a graph whose connected components each have at most~$\iota - \iota'$ vertices.
Many problems can be solved in~$O(nm)$~time and thus in $O(\iota n^2)$ time, since $m \in O(\iota n)$.
As the relation $\td \le \iota \le \vc + 1$ holds for any graph,
an algorithm that runs in $O(m  \cdot \td)$ time (e.g., for \textsc{Maximum Matching}) also runs in~$O(\iota^2 n)$~time.
These bounds become $O(n^3)$ when~$\iota = \Theta(n)$.
However, many problems can be solved in faster~$O(n^{\omega})$~time using fast matrix multiplication.
We aim to close this gap by developing fully polynomial-time algorithms that run in $O(n^{\omega})$ time even when~$\iota = \Theta(n)$.
Such algorithms are called \emph{adaptive} and are optimal unless there is an (unparameterized) algorithm that runs faster than $O(n^{\omega})$ time.
For many problems, the discovery of such an algorithm would be a breakthrough, given that these~$O(n^{\omega})$-time algorithms were developed decades ago and have not been improved since.

\subparagraph*{Our approach.}
Before describing our results, let us briefly discuss our approach (see \Cref{sec:prelim} for details).
Let $S$ be a~\emph{$k$-separator}, a vertex set of size~$k' \le k$ such that each connected component of~$G - S$ has size at most~$k - k'$.
Throughout the paper, we will make the assumption that a \mbox{$k$-separator} is given as input.
We remark that our algorithms do not require an (optimal)~\mbox{$\iota$-separator} to compute the correct solution.
However, the running times of our algorithm will depend on~$k$ and to achieve the claimed running times, we require a~$k$-separator with~$k \in O(\iota)$.
Although $G - S$ may have~$O (n)$ connected components, we may assume that there are $\Theta (n/k)$ ``components'' (which are not necessarily connected; see \Cref{sec:prelim} for details), each with~$O(k)$ vertices.
For every component~$C$, we can use the $O(n^{\omega})$-time algorithm to solve the instance on~$G[C]$ or~$G[S \cup C]$, which takes~$O(k^{\omega} \cdot n/k) = O(k^{\omega - 1} n)$ time. 
The next step is to combine solutions for $O(n / k)$ instances, which varies depending on the problem.
For instance, this is trivial for the problem of finding a triangle, as a triangle must be contained in $S \cup C$ for some component $C$.
For other problems, e.g., finding a maximum matching, this step requires a more sophisticated approach.

\subparagraph*{Main results.}
There are four main results in this work.
The first result concerns the problem of computing the girth of a graph, that is, the size of a smallest cycle.
The second result is about finding an induced copy of a graph $H$.
Vassilevska Williams et al.~\cite{DBLP:conf/soda/WilliamsWWY15} gave an $O(n^{\omega})$-time algorithm that finds an induced copy of $H$ when $H$ is a graph on four vertices that is not a clique $K_4$ or its complement $\overline{K_4}$.
Their randomized algorithm is based on computing the number of induced copies of $H$ modulo some integer $q$ which
they show to be computable from $A^2$ in linear time, where $A$ is the adjacency matrix.
We observe that the ``essential'' part of $A^2$ can be computed in~$O(\iota^{\omega - 1} n)$~time, leading to~$O(\iota^{\omega - 1} n)$-time randomized algorithms.

Third, we develop a randomized~$O(\iota^{\omega - 1} n)$-time algorithm for finding a maximum matching.
We start by showing that whether a graph contains a perfect matching can be determined in~$O(\iota^{\omega - 1} n)$~time.
Tutte \cite{tutte1947factorization} observed that the Tutte matrix is nonsingular if and only if the graph has a perfect matching.
By the Schwartz-Zippel lemma, we can test its nonsingularity in randomized~$O(n^{\omega})$~time.
We can thus test whether each component in $G - S$ has a perfect matching in~$O(\iota^{\omega - 1} n)$~time.
However, there might be a vertex in $S$ that must be matched to a vertex in~$G - S$.
To handle these cases, we use \emph{Schur complements}.
The task of finding a maximum matching is more intricate.
Lov\'{a}sz \cite{DBLP:conf/fct/Lovasz79} generalized Tutte's observation by stating that the rank of the Tutte matrix (which can be computed in randomized $O(n^\omega)$ time) equals twice the size of a maximum matching.
It was only decades later that $O(n^{\omega})$-time algorithms for finding one were discovered.
Mucha and Sankowski~\cite{DBLP:conf/focs/MuchaS04} and Harvey \cite{DBLP:journals/siamcomp/Harvey09} gave such algorithms.
We show how to adapt the latter to obtain a randomized~$O(\iota^{\omega - 1} n)$-time algorithm for finding a maximum matching.

Lastly, we study APSP on unweighted graphs.
Seidel \cite{DBLP:conf/stoc/Seidel92} showed that APSP can be solved in~$O(n^{\omega} \log n)$~time.
Alon and Yuster \cite{DBLP:conf/esa/AlonY07} later developed an algorithm that runs in~$O(\vc^{\omega - 2} n^2)$~time (they actually provide a stronger bound on the running time using rectangular matrix multiplication).
We show that APSP can be solved in~$O(\iota^{\omega - 2} n^2)$~time when the graph has constant diameter.
We were not able to obtain an adaptive algorithm in general, but we give an~$O(\iota^{(\omega-1)/2}n^2) \subseteq O(\iota^{0.687} n^2)$-time algorithm.
When parameterizing by $\vc$, we can effectively replace every vertex not in the vertex cover with edges of weight two connecting their neighbors.
Thus, the $O(Wn^{\omega})$-time algorithm~\mbox{\cite{DBLP:journals/jgaa/EirinakisWS17,DBLP:conf/focs/ShoshanZ99}} for weighted APSP, where~$W$~is the maximum weight, finds all pairwise distances between vertices in the vertex cover in~$O(\vc^{\omega})$~time.
For vertex integrity, we show how to replace every component with edges of weight~$O(\iota)$.
To compute distances between pairs of vertices with at least one vertex not in the~$k$-separator, we use the known subcubic-time algorithm for computing min-plus matrix multiplication for \emph{bounded-difference matrices}, matrices in which the difference of two adjacent entries in a row is constant \cite{DBLP:conf/stoc/ChiDX022}.

\subparagraph*{Previous work on vertex integrity.}
The notion of vertex integrity was introduced by Barefoot et al.~\cite{barefoot1987vulnerability}.
The vertex integrity $\iota$ can be much smaller than $n$, e.g., it is known that~$\iota \in O(n^{2/3})$ on $K_h$-minor free graphs \cite{DBLP:journals/siamdm/BenkoEL09}.
\textsc{Vertex Integrity}, i.e., the problem of computing an $\iota$-separator is NP-hard.
A straightforward branching algorithm solves \textsc{Vertex Integrity} in~$O(\iota^{\iota} \cdot n)$ time (see~\cite{DBLP:journals/algorithmica/DrangeDH16}).
A greedy algorithm can find an~$O(\iota^2)$-separator in linear time.
There is also a polynomial-time algorithm that can find an~$O(\iota \log \iota)$-separator~\cite{DBLP:journals/mp/Lee19}.
FPT algorithms parameterized by vertex integrity gained increased attention recently~\cite{DBLP:journals/algorithmica/BodlaenderHKKOO20,DBLP:journals/algorithmica/DrangeDH16,DBLP:journals/ai/DvorakEGKO21,DBLP:journals/tcs/GimaHKKO22,DBLP:conf/isaac/LampisM21}.
In particular, see Gima et al.~\cite{DBLP:journals/tcs/GimaHKKO22} for an extensive list of problems that are W[1]-hard for tree-depth but become FPT when parameterized by vertex integrity.

\section{Preliminaries}
\label{sec:prelim}

We use standard notation from graph theory.
Unless stated otherwise, all appearing graphs are undirected.
Further, $V$ denotes the set of vertices in the graph, $E$ its set of edges, $n$ its number of vertices, and $m$ its number of edges.
We denote an edge between two vertices $u$ and $v$ by~$uv$.
A \emph{walk} of length~$\ell$ is a sequence~$v_1, \dots, v_{\ell + 1}$ of (not necessarily distinct) vertices such that $v_i v_{i+1} \in E$ for all $i \in [\ell]$, where $[j]:= \{1, \dots, j\}$ for any integer~$j$.
A walk whose vertices are all pairwise distinct is a \emph{path}.
The \emph{adjacency matrix} of~$G$ is the $V \times V$-matrix~$A$ with~$A [u, v] = 1$ if and only if $uv \in E$, and $A[u, v] = 0$ otherwise (where $A[u, v]$ is the entry of~$A$ indexed by~$u$ and $v$).

\subparagraph*{Graph parameters.} For a graph $G$, the \emph{vertex integrity} is the smallest integer~$\iota$ such that~$G$~contains a set $S$ (called \emph{$\iota$-separator}) of $\iota' \le \iota$ vertices whose deletion results in a graph whose connected components each contain at most $\iota - \iota'$ vertices.
The \emph{vertex cover number} $\vc$ is the smallest cardinality of a vertex cover, a set of vertices that contains at least one endpoint of every edge.
The \emph{tree-depth} $\td$ is the smallest depth of a rooted forest $F$ with vertex set $V$ such that~$G$~can be \emph{embedded} in $F$, i.e., for every edge $xy$ in $G$, $x$ is an ancestor of $y$ or vice versa.
The \emph{feedback vertex number} is the smallest cardinality of a feedback vertex set, a set of vertices that contains at least one vertex of every cycle.
The \emph{girth} is the size of a smallest cycle in a graph.

\subparagraph*{Decomposition.}

Here, we describe the decomposition with respect to a $k$-separator, which will be used throughout the paper.
Let $S$ be a $k$-separator.
Typically in our algorithms, we spend~$O(k^{\omega})$~time for every connected component in $G - S$.
Since $G - S$ may have~$\Omega(n)$ connected components, this would result in a running time of $O(k^{\omega} n)$, which is often worse than a more straightforward algorithm.
Thus, we will do the following to bound the number of ``components'' by $O(n / k)$:
Basically, we put together some connected components~$C$ and construct a collection~$\mathcal{T}$ of sets, each (except for possibly the last one) containing between~$k$ and~$2k-1$ vertices.
More precisely, we start with $\mathcal{T} = \emptyset$ and process the connected components of~$G-S$ one by one as follows.
If every set $T \in \mathcal{T}$ has at least $k$ vertices, then add $\{ C \}$ to $\mathcal{T}$, and otherwise replace the set $T \in \mathcal{T}$ with $|T| < k$ by $T \cup C$.
Since every connected component $C$ has at most $k$ vertices, every set $T \in \mathcal{T}$ (except for possibly the last set which may be smaller) contains between~$k$ and~$2k - 1$~vertices.
Let~$\mathcal{T} = \{ T_1, \dots, T_{\nu} \}$.
It is easy to see that $\ncc \le n / k + 1$.
In our algorithms, we will always assume that the decomposition $(S; T_1, \dots, T_{\nu})$ of $V$ is given.
Note that given a $k$-separator, the decomposition can be computed in linear time.

\subparagraph*{Basic operations in matrix multiplication time.}

For $n \times n$-matrices $A$ and~$B$, one can compute the following in $O(n^{\omega})$ time: (i) the product $AB$, (ii) the determinant $\det A$, (iii) the inverse~$A^{-1}$, and (iv) a row/column basis of $A$ (see for example~\cite{DBLP:books/daglib/0090316}).
More generally, for a~$k \times n$-matrix $A$ and an~$n \times k$-matrix $B$, one can compute the product~$AB$ in $O(k^{\omega - 1} n)$ time by dividing~$A$ and $B$ into $n / k$ blocks of size $k \times k$.
The rank of~$A$ can be computed using Gaussian elimination using $O(k^{\omega - 1} n)$ arithmetic operations \cite{bunch1974triangular}.
Throughout the paper, we will use a word RAM model with word size $O(\log n)$.
If $\mathbb{F}$ is a field of size $\poly(n)$, we will assume that addition and multiplication take $O(1)$ time.

\subparagraph*{Matrices and Matchings.}
For a subset~$X$ of rows and a subset~$Y$ of columns, we denote by~$A[X, Y]$ the restriction of the matrix~$A$ to rows~$X$ and columns~$Y$.
When $X$ contains all rows (or $Y$ all columns), we simplify this notation to $A[\cdot, Y]$ (or $A[X, \cdot]$ for columns).
If $X$ (or $Y$) is the entire row set (column set, respectively), then we write $A[\cdot, Y]$ ($A[X, \cdot]$, respectively).
For a set~$X$ of rows (or columns), we will use the shorthand $A[X]$ for $A[X, X]$.
The identity matrix is denoted by $I$, and the zero matrix is denoted by $\mathbf{0}$.
The $i$-th power of the adjacency matrix~$A$ correspond to the number of walks of length~$i$, i.e., $A^i [u, v] $ equals the number of $u$-$v$-walks of length exactly~$i$ in~$G$.

Note that for any graph with a~$k$-separator~$S$ and decomposition~$(S; T_1, \dots, T_\ncc)$, it holds that there is no edge between a vertex in~$T_i$ and a vertex in~$T_j$ for any~$i \neq j$.
We can therefore represent the adjacency matrix~$A$ of the graph as follows.
\begin{align}
    \label{eq:adj-matrix}
    A = \begin{blockarray}{ccc}
        & S & \bar{S} \\
        \begin{block}{c[cc]}
            S & \gamma & \beta \\
            \bar{S} & \beta^T & \alpha\\
        \end{block}
    \end{blockarray}, \text{ where }
    \alpha = \begin{blockarray}{cccc}
        & T_1 & \cdots & T_{\ncc} \\
        \begin{block}{c[ccc]}
            T_1 & \alpha_1 & \cdots & \bold{0} \\
            \vdots & \vdots & \ddots & \vdots \\
            T_{\ncc} & \bold{0} & \cdots & \alpha_{\ncc} \\
        \end{block}
    \end{blockarray} \text{ and }
    \beta = \begin{blockarray}{cccc}
        & T_1 & \cdots & T_{\ncc} \\
        \begin{block}{c[ccc]}
            S & \beta_1 & \cdots & \beta_{\ncc} \\
        \end{block}
    \end{blockarray}.
\end{align}
Here, $\bar{S} = T_1 \cup \dots \cup T_{\ncc}$ and the matrices $\alpha_i, \beta_i, \gamma$ all have size $O(k) \times O(k)$.
Many of our algorithms will use this representation and exploit the sparseness when computing e.g.\ matrix multiplications and determinants.

We denote the (unique) finite field with $2^q$ many elements by $\GF(2^q)$.
Note that this field has characteristic~2, i.e., $x + x = 0$ for every $x\in \GF(2^q)$.
For a graph $G = (V, E)$, the Tutte matrix (also known as the skew adjacency matrix) $A$ whose rows and columns are indexed by~$V = \{v_1, \dots, v_n\}$ is defined by 
\begin{align*}
    A[u, v] = \begin{cases}
        +x_{uv} & \text{ if } u = v_i, v = v_j \text{ with } i < j \text{ and } uv \in E \\
        -x_{uv} & \text{ if } u = v_i, v = v_j \text{ with } j < i \text{ and } uv \in E \\
        0 & \text{ otherwise}, \\
    \end{cases}
\end{align*}
where $x_{uv}$ is a variable associated with the edge $uv$.
The Tutte matrix $A$ is \emph{skew-symmetric}, that is,~$A = -A^T$.
The Pfaffian of a skew-symmetric matrix~$A$ indexed by $V$ is defined as
\begin{align*}
    \pf(A) = \sum_{M \in \mathcal{M}} \sigma_M \prod_{uv \in M} A[u, v],
\end{align*}
where $\mathcal{M}$ is the set of all perfect matchings of $(V, \binom{V}{2})$ and $\sigma_M \in \{ \pm 1\}$ is the sign of $M$.
We will assume that the field has characteristic 2 (implying $-1=1$), and thus the precise definition of $\sigma_M$ is not important for us.
(This assumption is not essential to our algorithm but it will simplify the notation.)

The following are well-known facts about skew-symmetric matrices (see e.g., \cite{DBLP:books/daglib/0077284,murota1999matrices}).
\begin{lemma}
    \label{lemma:pfaffian}
    For a skew-symmetric matrix $A$, we have $\det A = \pf(A)^2$.
\end{lemma}
In particular, a skew-symmetric matrix $A$ is nonsingular if and only if $\pf(A) \ne 0$.
\begin{lemma}
    \label{lemma:ss-basis}
    For a skew-symmetric matrix $A$, if $X$ is a row (or column) basis, then $A[X]$ is nonsingular.
\end{lemma}

The next is immediate from the definition of the Pfaffian.
\begin{lemma}[row expansion] \label{lemma:laplace}
     For a skew-symmetric matrix $A$ indexed by $V$ and $v \in V$ over a field of characteristic 2, we have
     $\pf(A) = \sum_{v' \in V \setminus \{ v \}}A[v, v'] \cdot  \pf(\widehat{A}_{v, v'})$,
     where $\widehat{A}_{v, v'}$ is the matrix where the rows and columns indexed by $v$ and $v'$ are deleted.
\end{lemma}

\section{Finding Short Cycles}
\label{sec:girth}
In this section, we will show how to compute the girth of a graph in $O(k^{\omega - 1} n)$ time when a $k$-separator of the input graph is given.
Recall the \emph{girth} of a graph is the length of its shortest cycle.
Although the problem of finding the girth is seemingly more difficult than finding a triangle, these two problems are subcubic equivalent---in fact, there is an algorithm computing the girth in~$O(n^{\omega})$ time \cite{DBLP:journals/siamcomp/ItaiR78}.
We also give a more efficient algorithm to determine the \emph{even girth}. 
The algorithm for computing the girth uses the following~$O(n)$-time BFS-based procedure that terminates as soon as a cycle is encountered.

\begin{lemma}[\cite{DBLP:journals/siamdm/YusterZ97}]
    \label{lemma:bfs-cycle}
    Let $G$ be a graph and $\ell$ be the girth of $G$.
    There is an $O(n)$-time algorithm that given a vertex $v \in V$, determines that $v$ is not on a shortest cycle or (i) finds a cycle of length $\ell$ if $\ell$ is even and (ii) finds a cycle of length of at most $\ell + 1$ and all vertices at distance $(\ell - 1) / 2$ from $v$ if $\ell$ is odd.
\end{lemma}

\Cref{lemma:bfs-cycle} helps us focus on the case where the girth is odd.

\begin{restatable}{proposition}{thmgirth}\label{thm:girth}
    \mbox{Given a graph and its $k$-separator, we can find its girth in $O(k^{\omega - 1} n)$ time.}
\end{restatable}

\begin{proof}
    First, we check whether~$G$ contains a triangle in~$O(k^{\omega -1} n)$ time.
    If so, then this triangle is clearly a shortest cycle.
    Thus, we assume in the following that~$G$ is triangle-free.

    We compute the girth of the graph $G[T_i]$ for every $i$, using the $O(n^\omega)$-algorithm from Itai and Rodeh~\cite{DBLP:journals/siamcomp/ItaiR78}.
    This takes $O(\ncc \cdot k^\omega) = (k^{\omega - 1} n)$ time.
    This covers the case that the shortest cycle is fully contained in $T_i$ for some $i \in [\ncc]$.

    Now we may assume that a shortest cycle intersects $S$. 
    For each vertex~$v \in S$, we run the algorithm of \Cref{lemma:bfs-cycle}.
    Let~$\ell_v$ the length of the cycle found when calling \Cref{lemma:bfs-cycle} for~$v \in S$, and let~$\ell_{\min} := \min_{v \in S} \ell_v$.
    If $\ell_{\min}$ is odd, then the corresponding cycle is already a shortest cycle in~$G$ by \Cref{lemma:bfs-cycle}.
    So assume that $\ell_{\min} $ is even.
    Itai and Rodeh~\cite{DBLP:journals/siamcomp/ItaiR78} reduced the problem of determining whether a shortest cycle in~$G$ has length $\ell_{\min}$ or $\ell_{\min} - 1$ to the problem of finding a triangle in an auxiliary graph~$H$.
    The graph~$H$ is constructed as follows, starting with a copy of~$G$:
    Let~$L:= \{v \in S \mid \ell_v = \ell_{\min}\}$.
    Then, for any~$v \in L$, we add a vertex~$v^*$ to~$H$ and connect~$v^*$ to all vertices in~$G$ which have distance exactly $(\ell_{\min}-1)/2$ from~$v$ (note that the algorithm from \Cref{lemma:bfs-cycle} computes these vertices).
    Note that $S \cup \{v^* \mid v\in L\}$ is a $2k$-separator.
    Consequently, finding a triangle in~$H$ can be done in $O(k^{\omega-1}n)$ time.
    Thus, the whole algorithm runs in~$O(k^{\omega-1}n)$~time.
\end{proof}

We continue with the even girth, that is, the length of a shortest even cycle.

\begin{restatable}{proposition}{propevencycle}\label{prop:even-cycle}
    Given a graph~$G$ and its $k$-separator~$S$, a shortest even cycle can be found in~$O(kn \alpha (n))$~time (if one exists), where $\alpha$ is the inverse Ackermann function.
\end{restatable}
\begin{proof}
    We make a case distinction whether there is a shortest even cycle containing a vertex from~$S$ or not:
    If every shortest even cycle is disjoint from~$S$, then it suffices to check~$T_1, \dots, T_\ncc$ for a shortest even cycle.
    For each~$i \in [\ncc]$, a shortest even cycle can be found in $O(|T_i|^2) = O(k^2)$ by the algorithm from \citet{DBLP:journals/siamdm/YusterZ97}.
    As $\ncc \le n/k$, this case runs in $O(nk)$ time overall.

    It remains to consider the case that every shortest even cycle contains at least one vertex from~$S$.
    Monien~\cite{DBLP:journals/computing/Monien83} showed how to compute a shortest even cycle containing a given vertex~$v$ in~$O(n \alpha (n))$~time.
    For each vertex~$s \in S$, we compute the shortest even cycle containing~$s$ in~$O(n\alpha (s))$~time using Monien's algorithm.
    The running time spent for all vertices from~$S$ together is~$O(kn \alpha (n))$, implying the theorem.
\end{proof}

We conclude this section by showing that finding a (not necessarily induced) cycle $C_{\ell}$ of length~$\ell$ for constant~$\ell$ takes~$O(k^{\omega - 1} n)$~time.
This is an adaptation of the color-coding algorithm of Alon et al.~\cite{DBLP:journals/jacm/AlonYZ95}.

\begin{restatable}{proposition}{thmconstcycle}\label{thm:constant-cycle}
 Given a graph~$G$ and its $k$-separator~$S$, a cycle $C_{\ell}$ of length $\ell$ can be found in~$O(k^{\omega -1} n)$~time for any fixed $\ell$ if one exists.
\end{restatable}

\begin{proof}
 First, we check whether~$T_i$ contains an $\ell$-cycle for some~$i$ in~$O (k^\omega)$ time~\cite{DBLP:journals/jacm/AlonYZ95}.
 This takes overall~${O(k^\omega \ncc) = O( k^{\omega -1} n)}$~time.
 If we find such a cycle, then we are done;
 so assume that each~$C_{\ell}$ contains at least one vertex from~$S$.
 
 We use color-coding and assign each vertex a color~$c(v)$ chosen from~$[\ell]$ uniformly at random.
 We assume without loss of generality that there is some $C_{\ell}$ containing a vertex from~$S$ colored with~1.
 From this coloring, we construct a directed graph~$H$ in which the vertex set is $V$ and there is an arc from $v$ to $w$ if and only if $vw \in E$ and $c(w) - c(v) \equiv 1 \mod \ell$.
 For each~$i \in [\ncc]$, let~$M_i$ be the adjacency matrix of $H[S \cup T_i]$.
 For every $u, v \in S$ and $\ell' \in [\ell]$, $M_i^{\ell'} [u, v] \neq 0$ if and only if there is a (directed) path of length exactly $\ell'$ from~$u$ to~$v$ in $H[S \cup T_i]$.
 We construct a graph~$H^*$ with vertex set~$S \cup S_1$ with~$S_1 = \{s_1 \mid s\in S \land c(s) = 1\}$ and an arc from $v$ to $w$ if $c (v) < c (w)$ and there is some~$i $ such that $M_i^{c(w) - c(v)} (v, w) \neq 0$.
 Further, $H^*$ contains an arc from $u$ to $s_1$ if there is some~$i$ such that~$M_i^{\ell - c(u)} (u, s) \neq 0$.
 We claim that $G$ contains a cycle~$C$ of length $\ell$ containing a vertex from~$S_1$ and obeying the color coding (i.e., the $i$-th vertex of~$C$ has color~$i$) if and only if~$H^*$~contains an~$s$-$s_1$-path for some~$s \in S_1$.
 Since we can decide in~$O (k^\omega) $ time which vertex is reachable from which other vertex in~$H^*$, we can decide the existence of an $s$-$s_1$-path in~$O (k^\omega) \subseteq O(k^{\omega -1} n)$ time.
 
 So assume that~$G$ contains a cycle~$C$ of length $\ell$ containing a vertex~$s^1 \in S_1$ and obeying the color coding.
 Let~$C = \langle s^1,P^1, s^2,P^2,s^3, \dots, s^j,P^j, s^1\rangle$ where $s^i \in S$ (but not necessarily $c(s^i) = i$) and the subpath~$P^i $ does not contain any vertex from~$S$.
 Clearly, all vertices of~$P^i$ are contained in some component~$T_{p_i}$.
 Hence, we have that~$M^{\ell -1}_{p_i} (s^i, s^{i+1}) \neq 0$ and $M^{\ell -1}_{p_j} (s^j, s^1) \neq 0$.
 It follows that $\langle s^1,s^2,s^3, \dots, s^j, s^1_1\rangle$ is an $s^1$-$s^1_1$-path in~$H^*$.
 
 For the reverse direction, let~$P = \langle s^1, s^2, s^3, \dots, s^1_1\rangle$ be an~$s^1$-$s^1_1$-path in~$H^*$.
 By the definition of~$H^*$, for each~$q \in [j]$, there exists some~$i_q$ such that $M_{i_q} ^{\ell-1} (s^q, s^{q+1}) \neq 0$ (or $M_{i_j}^{\ell -1} (s^j, s^1) \neq 0$ for~$j = q$).
 Consequently, there exists some~$s^q$-$s^{q+1}$-path~$P_q$ in~$S \cup T_{i_q}$.
 Note that by the definition of~$H$, the internal vertices of~$P_q$ have colors~$c(s^q) +1, c (s^q) + 2, \dots,  c(s^{q+1})-1$.
 By the definition of~$H^*$, we have that~$c (s^q) < c (s^{q+1})$.
 Consequently, the walk~$\langle s^1, P^1, s^2, P^2, \dots, s^1_1\rangle$ is a cycle of length~$\ell $ obeying the color coding.
\end{proof}

\section{Finding Small Induced Subgraphs}
\label{sec:subgraph}

In this section, we develop adaptive algorithms for finding four-vertex subgraphs.
There are eleven non-isomorphic graphs with 4 vertices:
the clique on four vertices ($K_4$) and its complement ($\overline{K_4}$), the diamond ($K_4 - e$) and its complement ($\overline{K_4} + e$), the claw ($K_{1, 3}$) and its complement ($\overline{K_{1, 3}}$), the paw ($K_{1,3} + e$) and its complement ($\overline{K_{1,3}} -e$), the cycle on four vertices ($C_4$) and its complement~($\overline{C_4}$), and the path on four vertices ($P_4$) (which is its own complement).
Here, $+e$ and $-e$ indicate the insertion of an edge or the deletion of any edge, respectively.
It is known that an induced~$P_4$ can be detected in linear time \cite{DBLP:journals/siamcomp/CorneilPS85}.
For $K_4$ and $\overline{K_4}$, the currently fastest algorithm runs in~$O(n^{3.257})$~\cite{DBLP:journals/tcs/EisenbrandG04,DBLP:conf/focs/Gall12} and for all other graphs, the best known algorithms run in~$O(n^\omega)$ time.
These running times are achieved by an algorithm by  Williams et al.~\cite{DBLP:conf/soda/WilliamsWWY15}.\footnote{We mention that there are a couple of four-vertex graphs for which algorithms with this running time were previously known, but the algorithm by Williams et al. gives a nice unifying approach. For this reason, we focus on their algorithm.}
The approach by Williams et al. for all graphs except the~$C_4$ and its compliment can be summarized as follows.
\begin{itemize}
    \item 
Let $G=(V,E)$ be an undirected graph and let~$A$ be its adjacency matrix.
Let $H = (V', E')$ be a four-vertex graph that is none of $K_4$, $\overline{K_4}$, $C_4$, $\overline{C_4}$.
There is an integer~$2 \leq q_H \leq 6$ such that if we can compute~$A^2[u,v]$ for every edge $uv \in E$ in time~$t$, then we can compute the number of induced copies of $H$ in $G$ modulo $q_H$ in~$O(n+m+t)$ time.
See~\cite[Lemma~4.1]{DBLP:conf/soda/WilliamsWWY15} for details.
(Some equations provided in \cite{DBLP:conf/soda/WilliamsWWY15} require $A^3[v]$ for every $v \in V$. However, this can be computed in~$O(m)$ time if~$A^2[u, v]$ is given for every edge $uv$.)

\item
Let~$q\geq 2$ be an integer and let~$G,H$ be two undirected graphs.
Let~$G'$  be an induced subgraph of $G$ obtained by independently deleting each vertex with probability $1/2$.
If~$G$ contains~$H=(V_H,E_H)$ as an induced subgraph, then the number of induced copies of~$H$ in~$G'$ modulo~$q$ is not 0 with probability at least~$2^{-|V_H|}$.
(For our applications, $|V_H| = 4$, so this probability is at least $1/16$.)
\end{itemize}

We show that when a $k$-separator is given, one can test in $O(k^{\omega -1} n)$ time whether there is an induced copy of $H$ for each four-vertex graph $H$ except for $K_4$ and $\overline{K_4}$.
We start with all graphs except for $K_4$, $\overline{K_4}$, $C_4$, $\overline{C_4}$.
Using the framework by Williams et al.~\cite{DBLP:conf/soda/WilliamsWWY15}, it suffices to show how to compute~$A^2[u,v]$ for every edge~$\{u,v\}$.
Clearly, it requires $\Omega(n^2)$ time to compute the square~$A^2$.
Our key observation is that the relevant part of~$A^2$ can be computed in~$O(k^{\omega-1} \cdot n)$ time.
(Incidentally, $A^2$ can be computed in $O(k^{\omega - 2} n^2)$ time; see \Cref{lem:multiplication-with-arbitrary-matrix}.)

\begin{lemma}
    \label{adjacencyMatrix}
    Given a graph~$G=(V,E)$ and a $k$-separator $S$ for~$G$, we can compute~$A^2[u,v]$ for every edge~${uv \in E}$ in~$O(k^{\omega-1} n)$ time.
\end{lemma}

\begin{proof}
    We use the decomposition $(S; T_1, \ldots, T_{\ncc})$ described in \Cref{sec:prelim} and suppose that the adjacency matrix $A$ has the form given in \Cref{eq:adj-matrix}.
Note that

\begin{minipage}{.3\textwidth}
\begin{align*}
    A^2 = \begin{blockarray}{ccccc}
        & S & T_1 & \cdots & T_{\ncc} \\
        \begin{block}{c[cccc]}
            S & \zeta & \eta_1 & \cdots & \eta_\ncc \\
            T_1 & \eta_1^T & \delta_1 & \cdots & - \\
            \vdots & \vdots & \vdots & \ddots & \vdots \\
            T_{\ncc} & \eta_\ncc^T & - & \cdots & \delta_{\ncc} \\
        \end{block}
    \end{blockarray}, \quad \text{ where}
\end{align*}
\end{minipage}
\begin{minipage}{.3\textwidth}
    \begin{align*}
    \zeta &= \gamma^2 + \beta \beta^T,\\
    \eta_i &= \gamma  \beta_i + \beta_i  \alpha_i, \text{ and }\\
    \delta_i &= \beta_i^T  \beta_i + \alpha_i^2.
\end{align*}
\end{minipage}

Note that computing $\zeta$ takes $O(\ncc \cdot k^\omega) = O(k^{\omega-1}n)$ time and computing each of the $O(\ncc)$ submatrices $\eta_i$ or $\delta_i$ takes $O(k^\omega)$ time.
The $-$ represents pairs where the corresponding vertices belong to different~$T_i$ and are therefore non-adjacent.
We thus do not need to compute these values.
Thus, we can compute all relevant values in $O(k^{\omega-1}n)$ time.
\end{proof}

By \Cref{adjacencyMatrix}, an induced copy of $H \notin \{K_4, \overline{K_4}, C_4, \overline{C_4}\}$ can be detected in~$O(\iota^{\omega - 1} n)$~time by a randomized algorithm.
We show next that $C_4$ and $\overline{C_4}$ can also be detected in~$O(\iota^{\omega - 1} n)$ (deterministic) time.

\begin{restatable}{proposition}{inducedCycle}
    \label{inducedCycle}
    Given a graph $G=(V,E)$, a $k$-separator for~$G$, and a graph~$H \in \{C_4,\overline{C_4}\}$, we can test whether~$G$ contains~$H$ as an induced subgraph in~$O(k^{\omega-1}n)$~time.
\end{restatable}

\begin{proof}
    First, observe that if the adjacency matrix has the form of \Cref{eq:adj-matrix}, the number of common neighbors of $u, v \in S$ in $T_i$ (i.e., $|N(u) \cap N(v) \cap T_i|$) equals $(\beta_i \beta_i^T)[u, v]$ since $${(\beta_i \beta_i^T)[u, v] = \beta_i[u] \beta_i^T[v] = \sum_{w \notin S} A[u, w] A[v, w]}.$$
    Thus, we can compute these values in ${O(k^{\omega} \nu ) = O(k^{\omega - 1} n)}$ time.
    
    We start with $C_4$.
    We can check in $O(k^{\omega-1} n)$ time whether there exists a component~$T_i$ in~$G-S$ such that $G[S \cup T_i]$ contains a $C_4$ by running the $O(n^\omega)$-time algorithm on $G[S \cup T_i]$ (which contain $O(k)$ vertices each) for each $i \in [\ncc]$.
    It remains to check whether there is an induced $C_4$ containing vertices from different components.
    Note that in this case there need to be exactly two vertices in~$S$ which are not adjacent and have common neighbors in two different components $T_i$ and $T_j$.
    Since~$\beta_i \beta_i^T$ has been computed for all $i \in [\ncc]$, we can simply check in $O(k^2 \ncc) = O(kn)$ time whether there are two non-adjacent vertices $u, v$ in $S$ with $(\beta_i \beta_i^T)[u, v] \ne 0$ and $(\beta_j \beta_j^T)[u, v] \ne 0$ for $i < j \in [\ncc]$.

    We continue with $\overline{C_4}$, which is a graph with two edges not sharing an endpoint.
    If there are two components which each induces at least one edge, then there is a $\overline{C_4}$.
    Hence, we may assume that only $T_1$ induces an edge.
    We then check in $O(k^\omega \ncc)$ time whether there is a component $T_i$ such that $G[S \cup T_1 \cup T_i]$ contains a $\overline{C_4}$ using an $O(n^{\omega})$-time algorithm.
    Suppose that no $\overline{C_4}$ is detected. 
    Then, $G$ contains $\overline{C_4}$ only if there are two nonadjacent vertices $u,v \in S$ and two components $T_i$ and $T_j$ such that there are vertices $x \in T_i$ and $y \in T_j$ with $xu \in E$, $xv \notin E$, $yu \notin E$, $yv \in E$.
    We call $x$ and $y$ private neighbors of $u$ and $v$, respectively. 
    With $u, v \in S$ fixed, whether $u$ has a private neighbor in a component~$T_i$  can be checked in constant time by comparing $|N(u) \cap T_i|$ and $(\beta_i \beta_i^T)[u, v]$:
    Since $(\beta_i \beta_i^T)[u, v]$ is the number of common neighbors of $u$ and $v$ in $T_i$, vertex~$u$ has a private neighbor if and only if $|N(u) \cap T_i| > (\beta_i \beta_i^T)[u, v]$.
    Observe that there are $i, j\in [\ncc]$ such that $u$ has a private neighbor in~$T_i$ and $v$ has a private neighbor in~$T_j$ if and only if each of $u$ and $v$ has a component~$T_i$ with a private neighbor in it and there are at least two components in which $u$ or $v$ has a private neighbor.
    Thus, we can find private neighbors of $u$ and $v$ (if they exist) in overall $O(k^{\omega-1}n)$ time.
\end{proof}

Thus, we obtain the following.

\begin{theorem}\label{thm:subgraphs-main}
    Given a graph~$G$, a $k$-separator, and a graph $H$ with four vertices that is not $K_4$ or~$\overline{K_4}$, we can test whether $G$ contains an induced copy of $H$ in randomized $O(k^{\omega - 1} n)$ time.
\end{theorem}

We can also find an induced copy with a constant overhead using a standard self-reduction.

\begin{restatable}{corollary}{corfind}\label{cor:find}
    Given a graph, a $k$-separator, and a graph $H$ with four vertices that is not $K_4$ or~$\overline{K_4}$, we can find an induced copy of $H$ in $O(k^{\omega - 1} n)$ time if it exists.
\end{restatable}
\begin{proof}
    We use a self-reduction:
    First, we apply the detection algorithm from \Cref{thm:subgraphs-main} in such a way that the error probability is at most $1/4$.
    If this algorithm determines that there is no induced copy of~$H$ in~$G$, then our algorithm concludes so.
    Otherwise we know that there is an induced copy of~$H$ in~$G$ and employ the following algorithm to find one:
    We arbitrarily divide the vertex set in five parts $V_1, \dots, V_5$ of equal size.
    We apply the detection algorithm on $G - V_i$ in a round-robin manner for all $i\in [5]$ (and repeat it in the case of failure).
    As soon as $H$ is detected in $G - V_i$ for some $i \in [5]$, we recursively search the subgraph $G - V_i$.
    As there exists an induced copy of~$H$ in $G - V_i$ for some $i\in [5]$ and the detection algorithm has a constant error probability, it follows that the expected running time before the recursion is $O(k^{\omega - 1} n)$. 
    This yields the recurrence relation
    $T(n, k) = T(\frac{4}{5} n, k) + O(k^{\omega - 1}n)$ (note that the $k$-separator in~$G$ is also a $k$-separator in $G - V_i$) for the expected running time.
    Solving the recurrence, the expected running time for finding an induced copy of $H$ is $O(k^{\omega - 1} n)$.
    To guarantee a running time of $O(k^{\omega -1} n)$, we stop the finding algorithm after $4c k^{\omega -1}n$ many steps and (falsely) conclude that there is no $H$ in~$G$.
    
    Now we analyze the success probability.
    If the algorithm returns an induced copy of~$H$, then the algorithm is always correct.
    The algorithm may falsely report that there is no induced copy of~$H$ only in two cases:
    (i) the detection algorithm used in a first step was wrong, which happens with probability $1/4$, or 
    (ii) the finding algorithm might take more than $4c k^{\omega -1}n$ steps.
    By Markov's inequality, this happens with probability~$1/4$.
    Thus, the error probability of the algorithm is~$1/2$.
\end{proof}

Finally, let us also remark on the detection of cliques $K_{\ell}$ and independent sets $\overline{K_{\ell}}$.
Let~$t_{\ell}(n)$ be the time complexity of finding $K_{\ell}$ (or $\overline{K_{\ell}}$) on an $n$-vertex graph.
(It is known, for instance, that~$t_4(n) \in O(n^{3.257})$ \cite{DBLP:journals/tcs/EisenbrandG04} using fast rectangular matrix multiplication \cite{DBLP:conf/focs/Gall12}.)
Since any $K_{\ell}$ must be fully contained in $G[S \cup T_i]$ for some $i$, it can be detected in $O(n \cdot \nicefrac{t_{\ell}(\iota)}{\iota})$ time.
For the detection of $\overline{K_{\ell}}$, note that if $n / \iota \ge \ell$, then there is an independent set of size~$\ell$.
Thus, we may assume that~$n \le \iota \ell$.
For constant $\ell$, we can thus find $\overline{K_{\ell}}$ in~$O(t_{\ell}(\iota))$ time.

\section{Matching}
\label{sec:matching}

In this section, we show the following two results.

\begin{restatable}{theorem}{matching} \label{theorem:matching}
    Given a graph and a $k$-separator, we can find a maximum matching in randomized~$O(k^{\omega - 1} n)$ time.
\end{restatable}

\begin{restatable}{theorem}{matchingfvs} \label{theorem:matching-fvs}
    We can find a maximum matching in randomized $O(k^{\omega - 1} n \log (n / k))$ time, where~$k$ is the feedback vertex number of the input graph.
\end{restatable}

Note that the running time of \Cref{theorem:matching-fvs} has an additional $O(\log n)$ factor compared to the running time of \Cref{theorem:matching-fvs} when $k \in \Theta(n^{1 - \varepsilon})$ for some~$0 < \varepsilon \leq 1$.

As mentioned in the introduction, Lov\'{a}sz \cite{DBLP:conf/fct/Lovasz79} showed that the cardinality of a maximum matching can be determined in randomized~$O(n^{\omega})$ time.
Several decades later, two algorithms have been developed to find a maximum matching \cite{DBLP:journals/siamcomp/Harvey09,DBLP:conf/focs/MuchaS04}.
Tutte \cite{tutte1947factorization} showed that~$G$ has a perfect matching if and only if its symbolic Tutte matrix~$A$ is nonsingular.
Lov\'{a}sz \cite{DBLP:conf/fct/Lovasz79} later showed that the rank of $A$ equals twice the size of maximum matching.
To avoid computation over a multivariate polynomial ring, we will assume that each variable $x_{uv}$ is instantiated with an element chosen from a finite field $\mathbb{F}$ uniformly at random. (We will assume that $|\mathbb{F}| = \GF(2^{c \lceil \log n \rceil})$ for a sufficiently large constant $c > 0$.)
Since the determinant of the symbolic Tutte matrix has degree at most~$n$, by the Schwartz-Zippel lemma \cite{DBLP:journals/jacm/Schwartz80,zippel1979probabilistic}\footnote{The Schwartz-Zippel lemma states that a non-zero polynomial of total degree $d$ over a finite field $\mathbb{F}$ is non-zero with probability at least $1 - d / |\mathbb{F}|$ when evaluated at a coordinate chosen uniformly at random.}, if~$G$ has a perfect matching and $\mathbb{F}$ is of size at least $\delta n$, then $A$ is nonsingular with probability at least~$1 - 1 / \delta$.

The problem of finding a perfect matching efficiently is more intricate.
(To find a maximum matching, search for a perfect matching in $G[X]$, where $X$ is a basis of the Tutte matrix.)
For an edge $uv \in E(G)$, it can be shown that there is a perfect matching containing $uv$ if and only if $A^{-1}[u, v] \ne 0$ (see e.g., \cite{DBLP:journals/jal/RabinV89}).
This gives us an $O(n^{\omega + 1})$-time algorithm for finding a perfect matching.
It remained open whether there is an $O(n^{\omega})$-time algorithm for several decades.
It was resolved by Mucha and Sankowski~\cite{DBLP:conf/focs/MuchaS04}.
Harvey~\cite{DBLP:journals/siamcomp/Harvey09} later gave a simpler algorithm.

\paragraph*{Harvey's algorithm.}
As we build upon Harvey's algorithm to develop fully polynomial-time algorithms, we briefly describe how it works.
The Sherman--Morrison--Woodbury formula is fundamental to Harvey's algorithm.

\begin{lemma}[Sherman--Morrison--Woodbury formula]
  Let $A$ be an $n \times n$ matrix and $U, V$ be $n \times k$ matrices.
  If $A$ is nonsingular, then 
  \begin{itemize}
    \item $\det(A + UV^T) = \det(I + V^T A^{-1} U) \det A$; in particular, $A + UV^T$ is nonsingular if and only if $I + VA^{-1}U$ is nonsingular, and
    \item $(A + UV^T)^{-1} = A^{-1} - A^{-1} U (I + V A^{-1} U)^{-1} V A^{-1}$.
  \end{itemize}
\end{lemma}

Harvey showed the following based on this formula. 

\begin{lemma}[Harvey~\cite{DBLP:journals/siamcomp/Harvey09}] \label{lemma:harvey}
    For an $n \times n$ nonsingular matrix $M$ and $S, T \subseteq [n]$, let $\widetilde{M}$ be a matrix that is identical to $M$ except that $\Delta = \widetilde{M}[S,T] - M[S,T] \ne \mathbb{0}$ (i.e., $M[i,j] = \widetilde{M}[i, j]$ whenever $i \notin S$ or $j \notin T$).
    Then,
    \begin{itemize}
        \item $\det(\widetilde{M}) = \det(M) \det(I + \Delta M^{-1}[T, S])$; in particular, $\widetilde{M}$ is nonsingular if and only if~$I + \Delta M^{-1}[T, S]$ is nonsingular, and
        \item if $\widetilde{M}$ is nonsingular, then $\widetilde{M}^{-1} = M^{-1} - M^{-1}[\cdot, S] (I + \Delta M^{-1}[T, S])^{-1} \Delta M^{-1}[T, \cdot]$.
    \end{itemize}
\end{lemma}

The main idea in Harvey's algorithm is to delete edges until each vertex has exactly one neighbor while simultaneously ensuring that at least one perfect matching remains.
Here, deleting an edge corresponds to setting the corresponding two entries in the Tutte matrix to zero.
Whether an edge is \emph{deletable}, i.e., whether there remains a perfect matching after its deletion, can be determined in constant time, since it suffices to test whether some $2 \times 2$ matrix is nonsingular by \Cref{lemma:harvey}.
When an edge is deleted, we update the inverse using \Cref{lemma:harvey}.
A naive implementation, however, takes $O(n^2)$ time, resulting in an $O(n^2 m)$-time algorithm.
Harvey's algorithm employs a divide-and-conquer approach to obtain an $O(n^{\omega})$-time algorithm.
We will not cover the general case, but only the bipartite case.
Let $A$ be the Tutte matrix and $N$ be its inverse.
Harvey gave a recursive algorithm \textsc{DeleteEdgeCrossing} that takes a bipartition $(U, W)$ of equal size as input.
It maintains the invariant that~$B[U \cup W] = A^{-1}[U \cup W]$ upon invocation.
For the base case that $U = \{ u \}$ and $W = \{ w \}$, if the edge $uw$ exists and is deletable, then set $A[u,w] = 0$.
If~$|U| = |W| > 1$, then we divide $U$ and $W$ into $U = U_1 \cup U_2$ and $W = W_1 \cup W_2$ of equal size.
For each $i \in \{ 1, 2 \}$ and $j \in \{ 1, 2 \}$, we recurse into $(U_i, W_j)$.
After each recursion, we update~$B[U \cup W]$ using \Cref{lemma:harvey} in order to maintain the invariant. 
Note that this takes $O(|U|^{\omega})$ time, as we only update $B$ locally.
To find a perfect matching, run \textsc{DeleteEdgeCrossing} on the input bipartition.
The running time~$f(n)$ follows the recurrence $f(n) = 4 f(n/2) + O(n^{\omega})$, which gives~$f(n) = O(n^{\omega})$ (assuming $\omega > 2$).

\paragraph*{Organization.}
We present algorithms for detecting a perfect matching in \Cref{ss:detect-pm} and for constructing a perfect matching in \Cref{ss:find-pm}.
\Cref{ss:find-mm} discusses how to find a maximum matching, proving \Cref{theorem:matching}.
Finally, we develop an adaptive algorithm parameterized by the feedback vertex number in \Cref{ss:matching-fvs}.

\subsection{Detecting a perfect matching} \label{ss:detect-pm}

In this section, we show how to find a maximum matching in $O(k^{\omega - 1} n)$ time when a $k$-separator is given:

\begin{proposition}
	\label{prop:pm}
    Given a graph and a $k$-separator, we can test whether the graph has a perfect matching in randomized $O(k^{\omega - 1} n)$ time.
\end{proposition}

The key ingredient is the Schur complement.
Consider a block matrix of the form
\begin{align*}
    A =
    \begin{bmatrix}
          \alpha & \beta \\ \gamma & \delta
        \end{bmatrix}.
\end{align*}
If $\alpha$ is nonsingular, then the determinant of $A$ is $\det(A) = \det(\alpha) \cdot \det(C)$, where $C = \delta - \gamma \alpha^{-1} \beta$ is the Schur complement (see e.g.~\cite{murota1999matrices}).
Thus, assuming that~$\alpha$ is nonsingular, $A$ is nonsingular if and only if $\delta - \gamma \alpha^{-1} \beta$ is nonsingular.
A simple application of the Schur complement yields the following.
\begin{lemma} \label{lemma:schur}
Consider a matrix~$A$ of the following form.
Then, provided that~$\alpha$ is nonsingular, $A$~is nonsingular if and only if the following matrix $A'$ is nonsingular.
\begin{center}
\begin{minipage}{.3\textwidth}
    \begin{align*}
    A =
    \begin{bmatrix}
          \alpha & \beta & \bold{0} \\
          -\beta^T & \gamma & \zeta \\
          \bold{0} & -\zeta^T & \eta
    \end{bmatrix}
\end{align*}
\end{minipage}
\begin{minipage}{.6\textwidth}
    \begin{equation*}
    A' = \begin{bmatrix}
        \gamma & \zeta \\
        -\zeta^T & \eta
    \end{bmatrix}
    -
    \begin{bmatrix}
        -\beta^T \\ \bold{0}
    \end{bmatrix}
    \alpha^{-1}
    \begin{bmatrix}
        \beta & \bold{0}
    \end{bmatrix}
    =
    \begin{bmatrix}
        \gamma + \beta^T \alpha^{-1} \beta & \zeta \\
        -\zeta^T & \eta
    \end{bmatrix}
    \end{equation*}
\end{minipage}
\end{center}
\end{lemma}

We prove \Cref{prop:pm} using \Cref{lemma:schur}.

\begin{proof}[Proof of \Cref{prop:pm}]
We assume that the Tutte matrix $A$ has the form of \Cref{eq:adj-matrix} and that it is instantiated randomly from a field~$\mathbb{F} = \GF(2^{c \lceil \log n \rceil})$.
For each component $T_i$, we find a basis~$T_i' \subseteq T_i$ of $A[T_i]$ in~$O(k^\omega)$ time.
Note that $A[T_i']$ is nonsingular by \Cref{lemma:ss-basis}.
Let $T_i^* = T_i \setminus T_i'$, let~$S_i = S_{i-1} \cup T^*_i$ with~$S_0 = S$, and let $S^* = S_\ncc$.
Note that if $G$ has a perfect matching $M$, then at least~$|T^*_i|$~vertices of~$T_i$ are matched to~$S$.
Thus, if $\sum_{i \in [\ncc]} |T^*_i| > |S|$, we can conclude that there is no perfect matching.
Otherwise, we have~$|S^*| \leq 2k$.
Now, consider for each component $T'_i$ the matrix
\begin{align*}
    B_i = \begin{blockarray}{cccc}
        & T'_i & S^* & V \setminus (S^* \cup \bigcup_{j \in [i - 1]} T'_j) \\
        \begin{block}{c[ccc]}
            T'_i & \alpha_{i} & \beta_i & 0 \\
            S^* & -\beta_{i}^T & \gamma_i & \zeta_i \\
            V \setminus (S^* \cup \bigcup_{j \in [i - 1]} T'_j) & 0 & -\zeta_i^T & \eta_i \\
        \end{block}
    \end{blockarray},
\end{align*}
where $\gamma_1 = \gamma$, $\gamma_{i + 1} =  \gamma_{i} + \beta_i^T \alpha_i^{-1} \beta_i$, and all entries except for~$\gamma_i$ for~$i > 1$ are identical to the corresponding entries of~$A$.
Note that $\alpha_i$ and $\beta_i$ may differ from $\alpha_i$ and $\beta_i$ in \Cref{eq:adj-matrix}, but if $T_j' = T_j$ for all $j\in [\nu]$, then $\alpha_i$ and $\beta_i$ here coincide with $\alpha_i$ and $\beta_i$ from \Cref{eq:adj-matrix}.
Note that the matrix~$B_1$ coincides with the Tutte matrix~$A$.
Hence, we only need to test whether $B_1$ is nonsingular.
We do this by iteratively computing~$B_{i+1}$ from $B_i$ in $O(k^\omega)$ time.
For notational convenience, we will assume that there is an empty component~$T'_{\ncc+1}$.
Since~$\alpha_i = A[T_i']$ is nonsingular, by \Cref{lemma:schur} the matrix $B_i$ is nonsingular if and only if~$B_{i+1}$ is nonsingular.
The nonsingularity of~$B_{\ncc+1} = \gamma_{\ncc +1}$ can be tested in~$O(k^{\omega})$ time since it has size~$O(k) \times O(k)$.
Note that our algorithm takes $O(k^{\omega})$ time for each~$i \in [\ncc+1]$ and thus~$O(\ncc k^{\omega}) = O(k^{\omega - 1}n)$ time overall.
If $A$ is nonsingular, then our algorithm correctly concludes that there is a perfect matching.
By the Schwartz--Zippel lemma, we know that if $G$ admits a perfect matching, then the probability that $A$ is nonsingular is at least~$1 - n / |\mathbb{F}| \ge 1/2$.
\end{proof}

\subsection{Finding a perfect matching} \label{ss:find-pm}

We next give an algorithm to find a perfect matching (if one exists).
We will use the same notation as before and we will again assume that $\alpha_i = A[T'_i]$ is nonsingular for every $i \in [\ncc]$.
Note that if the submatrix $\gamma_1$ is also nonsingular, then it is easy to find a perfect matching since the corresponding submatrices $\gamma_1 = A[S^*]$ and $\alpha_i$ are all nonsingular and therefore there exists a perfect matching in $G[S^*]$ and $G[T'_i]$ for every $i \in [\ncc]$.
We can find these using one of the randomized~\mbox{$O(n^{\omega})$-time} algorithms \cite{DBLP:journals/siamcomp/Harvey09,DBLP:conf/focs/MuchaS04}.
This procedure takes $O(k^{\omega} \ncc) = O(k^{\omega - 1} n)$ time.
However, in general, we cannot assume that $\gamma_1$ is nonsingular.

To deal with the case that $\gamma_1$ is singular, we use Harvey's aforementioned algorithm.
Our algorithm first computes $\gamma_i$ for each $i \in [\ncc + 1]$.
We then proceed inductively in decreasing order from $i = \ncc+1$ to $i = 1$.
Our algorithm finds a set $\widetilde{S}_{i} \subseteq \widetilde{S}_{i + 1} \subseteq S^*$ (we set~$\widetilde{S}_{\ncc + 1} = S^*$ for notational convenience) and a matching $M_i$ which saturates every vertex in $(\widetilde{S}_{i + 1} \setminus \widetilde{S}_i) \cup T'_{i}$.
The matching~$M_i$ is the union of two matchings~$M_i'$ and~$M_i''$, where $M_i'$ is a perfect matching between $\widetilde{S}_{i + 1} \setminus \widetilde{S}_i$ and~$T_i' \setminus T_i''$ for $T_i'' \subseteq T_i'$  and $M_i''$ is a perfect matching in $G[T_i'']$ (see \Cref{lemma:submatching,lemma:nonsingularity}).

We will maintain the invariant that~$\gamma_{i}[\widetilde{S}_{i}]$ is nonsingular.
If this invariant is maintained, then~$\gamma_1[\widetilde{S}_1]$ is nonsingular and thus $G[\widetilde{S}_1]$ contains a perfect matching that can be combined with all previous~$M_i$ to obtain a perfect matching for~$G$.
Note that~${\gamma_{\ncc+1}[\widetilde{S}_{\ncc+1}] = B_{\ncc+1}}$ is nonsingular.
To find the set~$\widetilde{S}_{i}$ for $i \in [\ncc]$, consider the following submatrix $B'_i$ of~$B_i$ whose rows and columns are indexed by $\widetilde{S}_{i + 1} \cup T'_i$:
\begin{align*}
    B'_i = \begin{bmatrix}
        \alpha_i & \beta_{i}[T'_i,\widetilde{S}_{i + 1}] \\
        -\beta_{i}^T[\widetilde{S}_{i + 1},T'_i] & \gamma_{i}[\widetilde{S}_{i + 1}] \\
    \end{bmatrix}.
\end{align*}

Note that for a nonzero entry in $B'_i$, if one of the coordinates is in $T'_i$, then the corresponding edge exists in $G$ but this is not necessarily the case if both coordinates are in $\widetilde{S}_{i + 1}$.
We first show that $B_i'$ is nonsingular.
\begin{restatable}{lemma}{schurnonsingularity} \label{lemma:schur-nonsingularity}
    For each $i \in [\ncc]$, the matrix $B'_i$ is nonsingular.
\end{restatable}

\begin{proof}
Recall that $\alpha_i$ is nonsingular and that by definition~$\gamma_{i+1} = \gamma_{i} + \beta_i^T \alpha_i^{-1} \beta_i$.
Using the Schur complement, we can express the determinant of $B'_i$ by
\begin{align*}
    \det(B'_i) = \det(\alpha_i) \det(\gamma_{i}[\widetilde{S}_{i + 1}] + \beta_i^T[\widetilde{S}_{i + 1},T'_i] \alpha_i^{-1} \beta_i[T'_i,\widetilde{S}_{i + 1}]) = \det(\alpha_i)\det(\gamma_{i+1}[\widetilde{S}_{i + 1}]).
\end{align*}
By our invariant, $\det(\gamma_{i + 1}[\widetilde{S}_{i+1}]) \ne 0$, and thus $B'_i$ is nonsingular.
\end{proof}

Since $B_i'$ is nonsingular, we can compute $(B_i')^{-1}$ and apply the subroutine \textsc{DeleteEdges-Crossing} of Harvey \cite{DBLP:journals/siamcomp/Harvey09} on the bipartition $(\widetilde{S}_{i + 1}, T'_i)$ in $O(k^{\omega})$ time.
(This subroutine requires the inverse to be given.)
Essentially, we ``delete'' (i.e., change to zero) all deletable entries with one coordinate in $\widetilde{S}_{i + 1}$ and the other in $T'_i$.
We then get a skew-symmetric matrix $C_i$ that is identical to $B'_i$ except that some entries in $C_i[\widetilde{S}_{i + 1}, T'_i]$ and $C_i[T'_i, \widetilde{S}_{i + 1}]$ are set to zero.
Let~$\widetilde{S}_{i}$ be those vertices~$v\in \widetilde{S}_{i + 1}$ such that~$C_i[T'_i,v]$ is a zero vector and~$T''_i$ be those vertices~$u\in T'_i$ such that~$C_i[\widetilde{S}_{i + 1},v]$ is a zero vector.
Each remaining nonzero entry in~$C_i[\widetilde{S}_{i + 1}, T'_i]$ and $C_i[T'_i, \widetilde{S}_{i + 1}]$ is undeletable.
We show that all edges corresponding to these entries form a perfect matching $M_i'$ between $\widetilde{S}_{i + 1} \setminus \widetilde{S}_i$ and $T'_i \setminus T''_i$.

\begin{lemma} \label{lemma:submatching}
    The set $M_i' = \{ uv \mid C_i[u, v] \ne 0, u \in \widetilde{S}_{i + 1} \setminus \widetilde{S}_i, v \in T'_i \setminus T''_i \}$ is a matching in~$G$ with high probability.
\end{lemma}
\begin{proof}
    For a vertex $u \in \widetilde{S}_{i + 1} \setminus \widetilde{S}_i$, assume that there are two vertices $v, v'$ such that~${C_i[u, v] \ne 0}$ and~$C_i[u, v'] \ne 0$.
    We show that this leads to a contradiction when the variables are treated as indeterminates.
    Let~$\widehat{C_i}_{u, w}$ be the result of deleting the rows and columns indexed by~$u$ and~$w$ from~$C_i$.
    By \Cref{lemma:laplace}, we have
    \begin{align} \label{eq:expansion}
        \pf(C_i) = C_i[u, v] \pf(\widehat{C_i}_{u, v}) + C_i[u, v'] \pf(\widehat{C_i}_{u, v'}) + p,
    \end{align}
    where $p = \sum_{\substack{w \in \widetilde{S}_{i+1} \cup T'_i,\, w \notin \{v, v'\}}} C_i[u, w] \pf(\widehat{C_i}_{u, w})$.
    
    By the assumption that deleting the corresponding edge $uv$ or $uv'$ (i.e., setting $x_{uv} = 0$ or~${x_{uv'} = 0}$)
    results in a singular matrix (this is due to Harvey's algorithm), we have
    \begin{align*}
    C_i[u, v] \pf(\widehat{C_i}_{u, v}) + p = 0 \text{ and } 
    C_i[u, v'] \pf(\widehat{C_i}_{u, v'}) + p = 0.
    \end{align*}
    We thus obtain $C_i[u, v] \pf(\widehat{C_i}_{u, v}) = C_i[u, v'] \pf(\widehat{C_i}_{u, v'}) = -p$.
    Note that $p \ne 0$, since otherwise~$\pf(C_i) = -p = 0$ by \Cref{eq:expansion}, which is a contradiction because $C_i$ is nonsingular by Harvey's algorithm.
    The left-hand side $C_i[u, v] \pf(\widehat{C_i}_{u, v})$ is multiplied by a variable~${C_i[u, v] = x_{u, v}}$, but this variable does not appear on the right-hand side~${C_i[u, v'] \pf(\widehat{C_i}_{u, v'})}$.
    Thus, we have that~$Q = C_i[u, v] \pf(\widehat{C_i}_{u, v}) - C_i[u, v'] \pf(\widehat{C_i}_{u, v'})$ is a polynomial that is not identically zero, which is a contradiction.
    We thus have exactly one vertex $v \in T'_i \setminus T''_i$ such that $C_i[u, v] \ne 0$.
    Since $Q$ has degree at most $n$, this evaluates to non-zero with high probability by the Schwartz--Zippel lemma.
    An analogous argument shows that for every~$v \in T'_i \setminus T''_i$, there is exactly one $u \in \widetilde{S}_{i + 1} \setminus \widetilde{S}_i$ such that $C_i[u, v] \ne 0$.
    Since every nonzero entry in $C_i[\widetilde{S}_{i + 1}, T_i]$ corresponds to an edge in $G$, we are done.
\end{proof}

We can now show that our main invariant can be maintained with high probability.

\begin{restatable}{lemma}{nonsingularity}\label{lemma:nonsingularity}
    The matrices $\gamma_{i}[\widetilde{S}_{i}]$ and $\alpha_i[T''_i]$ are nonsingular with high probability.
\end{restatable}

\begin{proof}
    Consider the monomial expansion of $\pf(C_i)$.
    Since plugging $x_{uv} = 0$ into $\pf(C_i)$ results in~0 for every nonzero entry $C_i[u, v]$ with $u \in \widetilde{S}_{i + 1} \setminus \widetilde{S}_i$ and $v \in T'_i \setminus T_i''$, every monomial in $\pf(C_i)$ is divisible by $x_{uv}$.
    In fact, every monomial is divisible by $\prod_{e \in M_i'} x_e$.
    Thus, in the summation~$\pf(C_i) = \sum_{M \in \mathcal{M}} \prod_{uv \in M} C_i[u, v]$, we can actually assume that $\mathcal{M}$ is the collection of matchings containing~$M_i'$, where $M_i'$ is a matching between $\widetilde{S}_{i + 1} \setminus \widetilde{S}_i$ and $T_i' \setminus T_i''$ in \Cref{lemma:submatching}.
    More formally, let~$\mathcal{M}_i$ be the set of all matchings in $(V',\binom{V'}{2})$, where~$V' = \widetilde{S}_i \cup T''_i$.
    Then,~$\pf(C_i) = \prod_{e \in M_i'} x_e \sum_{M \in \mathcal{M}_i} \prod_{uv \in M} C_i[u, v]$.
    Since the submatrix $C_i[\widetilde{S}_i, T''_i]$ is a zero matrix by definition, we can further assume that $\mathcal{M}_i$ only contains the product of all perfect matchings in~$\widetilde{S}_i$ and in~$T''_i$.
    Formally, let~$\mathcal{N}_i$ and $\mathcal{P}_i$ be the collection of all perfect matchings in complete graphs over vertex sets $\widetilde{S}_i$ and $T''_i$, respectively.
    Then,
    \begin{align*}
    \pf(C_i)
    = \prod_{e \in M_i'} x_e \sum_{M_1 \in \mathcal{N}_i, M_2 \in \mathcal{P}_i} \prod_{uv \in M_1} C_i[u, v] \prod_{uv \in M_2} C_i[u, v].
    = \prod_{e \in M_i'} x_e \cdot \pf(\gamma_{i}[\widetilde{S}_i]) \pf(\alpha_i[T''_i]),
    \end{align*}
    
    Since $\pf(C_i) \ne 0$, we have $\pf(\gamma_{i}[\widetilde{S}_i]) \ne 0$ and $\pf(\alpha[T''_i]) \ne 0$.
    Unless the subroutine \textsc{DeleteEdgesCrossing} fails, we will have $\pf(C_i) \ne 0$ in our algorithm.
    The failure probability is at most $\poly(n) / |\mathbb{F}|$ as argued by Harvey.
\end{proof}

We next show how to find a perfect matching.

\begin{restatable}{proposition}{findpm}
    \label{prop:find-pm}
    Given a graph with a perfect matching and a $k$-separator, we can find a perfect matching in $O(k^{\omega - 1}n)$ time.
\end{restatable}

\begin{proof}
    For each $i \in [\ncc]$, our algorithm finds $\widetilde{S_i} \subseteq \widetilde{S}_{i + 1}$ using the subroutine \textsc{DeleteEdgesCrossing}.
    There are perfect matchings $M_i'$ in $G[(\widetilde{S}_{i + 1} \setminus \widetilde{S_i}) \cup (T_i' \setminus T_i'')]$ and $M_i''$ in $G[T_i'']$ by \Cref{lemma:submatching,lemma:nonsingularity}, respectively.
    Harvey's algorithm finds $M_i''$ in $O(k^{\omega})$ time.
    By \Cref{lemma:nonsingularity}, the invariant that~$\gamma_i[\widetilde{S}_i]$ is nonsingular holds for every $i \in [\ncc + 1]$.
    In particular, $\gamma_1[\widetilde{S}_1] = \gamma[\widetilde{S}_1]$ is nonsingular, and thus Harvey's algorithm finds a perfect matching $M'$ in $G[\widetilde{S}_1]$ in $O(k^{\omega})$~time.
    Putting all these perfect matchings together yields a perfect matching for $G$.
    Our algorithm runs in~$O(\ncc k^{\omega}) = O(k^{\omega - 1} n)$~time.
    The failure probability is bounded by $\poly(n) / |\mathbb{F}|$, which is a constant when $\mathbb{F}$ is sufficiently large, e.g., $|\mathbb{F}| \in \Theta(n^5)$.
\end{proof}

\subsection{Finding a maximum matching} \label{ss:find-mm}
To find (the size of) a maximum matching when no perfect matching exists, we use the following standard approach.
It is known that the size of a maximum matching equals the size of a perfect matching in the graph $G[X]$, where~$X$ is a basis of the Tutte matrix \cite{DBLP:journals/siamcomp/Harvey09,DBLP:conf/focs/MuchaS04,DBLP:journals/jal/RabinV89}.
We show how to find a basis in $O(k^{\omega - 1} n)$ time.

We do the following for each $i \in [\ncc]$.
First, we compute a basis $T_i'$
of $\alpha_i[T_i]$ in~$O(k^\omega)$ time.
Note that we do not have the assumption that $\alpha_i$ is nonsingular.
We compute the Schur complement~$C$ of the block $A[\bigcup_{i \in [\nu]} T_i']$.
As in \Cref{lemma:schur}, taking the Schur complement of $A[\bigcup T_i']$ affects the submatrix $A[S \cup T_i]$.
For $T_i'' = T_i \setminus T_i'$, we thus have
\begin{align*}
    A[S \cup T_i] &= 
    \begin{blockarray}{cccc}
        & T_i' & T_i'' & S \\
        \begin{block}{c[ccc]}
            T_i' & \alpha_{i}[T_i'] & \alpha_i[T_i', T_i''] & \beta_i[T_i', \cdot] \\
            T_i'' & \alpha_i[T_i'', T_i'] & \bold{0} & \beta_i[T_i \setminus T_i', \cdot] \\
            S & \beta_i[\cdot, T_i'] & \beta_i[\cdot, T_i \setminus T_i'] & \gamma \\
        \end{block}
    \end{blockarray} \text{ and } \\
    C[S \cup T_i] &= \begin{blockarray}{ccc}
        & T_i' & T_i'' \\
        \begin{block}{c[cc]}
            T_i'' & -\alpha_i[T_i'', T_i'] (\alpha_i[T_i'])^{-1} \alpha_i[T_i, T_i''] & \beta_i' \\
            S & -(\beta_i')^T & \gamma' \\
        \end{block}
    \end{blockarray}
\end{align*}
where $\beta_i' = \beta_i[T_i \setminus T_i', \cdot] - \alpha_i[T_i', T_i''] (\alpha_i[T_i'])^{-1} \beta_i[T_i', \cdot]$.
By \Cref{lemma:schur}, it suffices to compute a basis in the Schur complement $C$.
We claim that $C[T_i'']$ is actually a zero matrix:
To see why, assume that there is an entry $C[u, v] = C[v, u] \ne 0$ with $u, v \in T_i$, i.e., $C[\{ u, v \}]$ is nonsingular.
We then have that $T_i' \cup \{ u, v \}$ is nonsingular in $A[T_i']$, contradicting that $T_i'$ is a basis of $A[T_i]$.
Thus, we have $A[T_i'', T_j''] = \mathbf{0}$ for all $i, j \in [\nu]$.
For a row basis $Y$ of $A[\bigcup_{i \in [\ncc]} T_i'', S]$, 
 we may assume that~$X \subseteq S \cup Y$.
Note that $Y$ can be computed in $O(k^{\omega - 1} n)$ time \cite{bunch1974triangular}.
Thus, we may delete all rows and columns in $\bigcup_{i \in [\ncc]} T_i''$ not part of $Y$.
Now the matrix $A$ has at most $2k$ rows and columns, and we can compute a basis of $A$ in~$O(k^\omega)$ time.
The whole procedure takes overall~$O(k^{\omega - 1} n)$~time and applying \Cref{prop:find-pm} to the computed basis yields \Cref{theorem:matching}.

\subsection{Parameterizing by feedback vertex number}
\label{ss:matching-fvs}

Recall that the feedback vertex number is the minimum number of vertices such that the graph becomes a forest after their deletion.
In this section, we prove the following:

\matchingfvs*

We start off by computing a feedback vertex set $S'$ of size at most $4\varphi$ in linear time \cite{DBLP:journals/siamcomp/Bar-YehudaGNR98} where $\varphi$ is the feedback vertex number.
By the same argument as in \Cref{ss:detect-pm}, we find a feedback vertex set~$S$ of size at most~$2|S'| \leq 8 \varphi$ such that~$G-S$ contains a perfect matching.
We will henceforth denote the size of $S$ by $k$.

Let $A$ be the Tutte matrix.
Our algorithm substitutes the variables $x_{uv}$ in $A$ with random elements from $\mathbb{F}$ and checks whether the resulting matrix $A$ is nonsingular.
To check whether $A$ is nonsingular, we would like to compute the Schur complement $A[S] - A[S, V \setminus S] B^{-1} A[V \setminus S, S]$, where $B := A[V \setminus S]$.
We show how to compute this in $O(k^{\omega - 1} n)$ time.
To that end, we use the algorithm of Fomin et al.~\cite{DBLP:journals/talg/FominLSPW18}, that solves a system of linear equations $Bx = b$ in $O(w^3 n)$ time, when the bipartite graph associated with $B$ has tree-width $w$.
In particular, the algorithm runs in~$O(n)$~time when the treewidth~$w$ is constant.
Using this algorithm $k$ times, $B^{-1} A[V \setminus S, S]$ can be computed in $O(kn)$ time.
With~$B^{-1} A[V \setminus S, S]$ at hand, we can compute the Schur complement in $O(k^{\omega - 1})$ time.
Thus, we have the following.

\begin{proposition}
    We can compute the size of a maximum matching in $O(k^{\omega - 1} n)$ time, where $k$ is the feedback vertex number.
\end{proposition}

Next, we discuss how to construct a perfect matching.
Let $\nu = n / k$.
We discuss how to find a perfect matching in~$O(k^{\omega} \nu \log \nu)$~time.
Our algorithm uses a divide and conquer approach following Harvey's algorithm.
We apply \textsc{DeleteEdgesCrossing} on $\nu$ matrices of dimension $O(k) \times O(k)$ to find all deletable edges $uv$ with $u \in S$ and $v \in V(G) \setminus S$.
The inverse needs to be updated after each call of \textsc{DeleteEdgeCrossing}.
To achieve this efficiently, the divide and conquer approach becomes essential.

The first step is to partially compute the inverse of the Tutte matrix.
We use the following formula to invert a block matrix:
\begin{align*}
    \begin{pmatrix}
        \alpha & \beta \\ \gamma & \delta
    \end{pmatrix}^{-1}
    = \begin{pmatrix}
        (\alpha - \beta \delta^{-1} \gamma)^{-1} & - (\alpha - \beta \delta^{-1} \gamma)^{-1} \beta \delta^{-1} \\ \delta^{-1} \gamma (\alpha - \beta \delta^{-1} \gamma)^{-1} & \delta^{-1} + \delta^{-1} \gamma (\alpha - \beta \delta^{-1} \gamma) \beta \delta^{-1}
    \end{pmatrix}
\end{align*}
Thus, we have $A^{-1}[S, V \setminus S] = -(A[S] - A[S, V \setminus S] B^{-1} A[V \setminus S, S])^{-1} A[S, V \setminus S] B^{-1}$, which can be computed in $O(k^{\omega - 1} n)$ time using the aforementioned algorithm of Fomin et al.~\cite{DBLP:journals/talg/FominLSPW18}.

We divide $V \setminus S$ into $\nu$ sets $T_1, \dots, T_{\nu}$ of equal size.
Let $B = A^{-1}[S, V \setminus S]$.
Our algorithm takes an interval $J \subseteq [\nu]$ as input.
We maintain the invariant that $B[\cdot, \bigcup_{i \in J} T_i] = A^{-1}[S, \bigcup_{i \in J} T_i]$ holds initially.
We give a recursive procedure \textsc{DeleteEdges} in \Cref{algorithm:divide-conquer}.
If $|J| = 1$, then we apply \textsc{DeleteEdgeCrossing} (from Harvey's algorithm) on $(S, T_i)$.
If $|J| > 1$, we partition $J$ into two sets $J_1$ and $J_2$, and recursively consider each.
After the first recursion, we update the inverse $B[S, \bigcup_{i \in J_2} T_i]$ as follows using \Cref{lemma:harvey} in order to maintain the invariant: 
\begin{align} \label{eq:myupdate}
    B[S, \bigcup_{i \in J_2} T_i] = B - B[S, S] (I + \Delta B[\bigcup_{i \in J_1} T_i, S])^{-1} \Delta S[T, \bigcup_{i \in J_2} T_i].
\end{align}
Here $\Delta$ is the difference in the Tutte matrix that occurs from edge deletions during the first recursion and it has size~$k \times |J_1| \cdot k$.
One can verify that this update can be performed in~${O(|J| \cdot k^{\omega})}$~time.
Let $f(\nu)$ denote the running time of \textsc{DeleteEdges}.
We have the recurrence
\begin{align*}
    f(\nu) = 2 f(\nu) + O(\nu k^{\omega}).
\end{align*}
Thus, we obtain $f(\nu) = O(\nu k^{\omega} \log \nu)$.

\begin{algorithm}[t]
\caption{An algorithm for finding a perfect matching}
\label{algorithm:divide-conquer}

\begin{algorithmic} 
\Procedure{DeleteEdges}{$I$}
    \State \textit{Invariant:} $B[\cdot, \bigcup_{i \in J} T_i] = A^{-1}[S, \bigcup_{i \in J} T_i]$.
    \If{$|J| = 1$}
        \State apply \Call{DeleteEdgeCrossing}{$S, T_i$} for $J = \{ i \}$
    \Else
        \State Partition $J$ into $J_1$ and $J_2$ of equal size
        \State \Call{DeleteEdges}{$J_1$}
        \State Update $B[S, \bigcup_{i \in J_2} T_i]$ using Equation~\eqref{eq:myupdate}
        \State \Call{DeleteEdges}{$J_2$}
    \EndIf
\EndProcedure
\end{algorithmic}
\end{algorithm}

It can be shown that non-zero entries in $A[S, V \setminus S]$ define a matching $M$ and that $A[S \setminus V(M)]$ and $A[(V \setminus S) \setminus V(M)]$ are nonsingular with high probability as we did in \Cref{lemma:submatching}.
We can find a perfect matching in $G[S \setminus V(M)]$ in $O(k^{\omega})$ time and a perfect matching in $G[(V \setminus S) \setminus V(M)]$ in~$O(n)$~time.

To find a maximum matching, it suffices to find a basis~$X$ of the Tutte matrix~$A$, and apply the algorithm for finding a perfect matching on~$G[X]$.
We find a maximum matching $M$ in $G[V \setminus S]$ in $O(n)$ time.
Let $T = V(M)$ and $T' = V \setminus (S \cup T)$.
Note that $A[T]$ is nonsingular.
We search for a basis in the Schur complement $A' = A[S \cup T'] - A[S \cup T', T] A[T]^{-1} A[T, S \cup T']$.
As we argued in \Cref{ss:find-mm}, we have $A'[T'] = \bold{0}$,
and thus a basis $Y$ in $A'$ can be found in $O(k^{\omega - 1} n)$ time.
Then, $Y \cup T$ forms a basis for $A$.
Thus, we have shown \Cref{theorem:matching-fvs}.

\section{All-Pairs Shortest Paths}
\label{sec:apsp}

	In this section, we study the well-known problem \textsc{All-Pairs Shortest Paths} (APSP) on unweighted graphs.
    APSP on unweighted graphs can be solved in~$O(n^{\omega})$-time using matrix multiplication~\cite{DBLP:conf/stoc/Seidel92}.
    We develop an algorithm which is faster than~$O(n^{\omega})$ when $k$ is sufficiently small, that is,~$k \in O(n^{2(\omega -2)/(\omega -1)}) \subseteq O(n^{0.541})$.

	\begin{restatable}{theorem}{apsp}
    	\label{thm:apsp}
		Given an (unweighted and undirected) graph~$G=(V,E)$ and a $k$-separator for~$G$, \textsc{APSP} can be solved in~$O(k^{(\omega-1)/2}n^2) \subseteq O(k^{0.687} n^2)$-time.
	\end{restatable}

    APSP is closely connected to the min-plus-product.
    In fact, solving APSP is subcubic equivalent to computing the min-plus product of two matrices \cite{DBLP:conf/focs/FischerM71}.
    The \emph{min-plus product}~$A \star B$ of a $p \times q$-matrix~$A$ with a $q \times r$-matrices $B$ is the $p \times r$-matrix $C$ with~${C[i,j] = \min_{k \in [q]} A[i, k] + B[k, j]}$.
    While it is conjectured that there is no algorithm computing the min-plus product of two $n \times n$-matrices in $O(n^{3-\epsilon})$ time for any~$\epsilon > 0$, there is a subcubic-time min-plus matrix multiplication algorithm for bounded-difference matrices, i.e., matrices in which the difference of two adjacent entries in a row is small \cite{DBLP:conf/stoc/ChiDX022}.
    Although the distance matrix $D$ (which contains pairwise distances) does not necessarily have bounded differences, for two adjacent vertices $u$ and $v$, the difference between $D[u, w]$ and $D[u, v]$ for any~$w \in V$ is at most~1.
    The ``standard'' decomposition of~$G - S $ into $T_1, \dots, T_\ncc$ from \Cref{sec:prelim} is not convenient in this case, since two vertices in~$T_i$ may have distance $\Theta(n)$, e.g., when $T_i$ is not connected.
    Our algorithm first modifies the given graph so that we have bounded-difference distance matrices. 
    The main idea is that, for a $k$-separator $S$, we can join connected components from~$G-S$ by adding copies of vertices in~$S$.
    We will also ensure that there is a Hamiltonian path through each component (this step is essentially also used by Deng et al.~\cite{DBLP:conf/icalp/DengKRWZ22}).
    This ensures that if we rearrange the rows and columns according to the Hamiltonian path, then two consecutive rows correspond to adjacent vertices and the distances between either of them and some third vertex differ by at most one.
    \begin{restatable}{lemma}{nicepart}\label{lem:apsp-nice-partition}
         Given a graph~$G=(V,E)$ and a $k$-separator~$S$ for~$G$, we can compute in $O(n+m)$ time a graph~$G'$ and a~$O(k)$-separator $S'$ for it such that $G' - S'$ has connected components~$T_1', \dots, T_\ncc'$ for $\ncc' \in O(n/k)$, each of which has a Hamiltonian path~$H_i$, and $|S'| = |T_1'| = \dots = |T_\ncc'|$.
         Moreover, given the distance matrix~$D'$ for~$G'$, one can compute the distance matrix~$D$ of $G$ in $O(n^2)$ time.
    \end{restatable}

\begin{proof}
            We assume without loss of generality that $G$ is connected; otherwise we can treat each connected component separately and then take the disjoint union over the graphs constructed for each connected component.
            We set~$S' := S$ and compute~$G'$ as follows:
            \begin{enumerate}
                \item For each connected component~$C$ of~$G - S$, choose one vertex~$s \in S$ that is adjacent to a vertex in~$C$.
                Let $f(C) := s$ and~$G' := G[S]$.
                \item Let $H=(S,E_H)$ be the graph with vertex set~$S$ which contains edge~$ss'$ if and only if $ss' \in E$ or there is a connected component~$C$ in $G- S$ such that there is an $s$-$s'$-path in~$G$ whose inner vertices are all contained in~$C$.
                Compute a rooted spanning tree~$T=(V_T,E_T)$ of~$H$ as follows:
                Pick an arbitrary vertex~$r$ from~$S$ to be the root of~$T$ and set~$R := \{r\}$ ($R$ will be the set of vertices currently reachable from~$r$).
                If there is an edge~$ss' \in E$ with~$s \in R $ and~$s' \in S\setminus R$, then add $s s'$ to~$T$ and set~$R := R \cup \{s'\}$.
                Otherwise pick a vertex~$s\in R$ such that there is some edge~$ss'\in E_H$ with~$s' \notin R$.
                Let $C$ be a component adjacent to both~$s$ and~$s'$.
                Add edge~$s s^*$ to~$T$ for each vertex~$s^* \in N_G (C) \setminus R$ and set $R := R \cup N_G (C)$.
                \item 
                Process the vertices from~$S$ in bottom-up manner.
                For each~$s \in S$, do the following:
                Let~$c_1, \ldots, c_\ell$ be the children of~$s$ and
                let $\mathcal{C}_s := f^{-1} (s)$ be the set of components for which the vertex~$s$ was chosen.
                Start with~$i = 1$ and $C_s^i = \emptyset$ for all $i$ ($C_s^i$ will be a connected component of~$G' - S'$).
                First, for each child $c_j$ of~$s$, we do the following:
                Let $C_{c_j}^{\max}$ be the last component created for~$c_j$ (i.e., $C_{c_j}^{i^*}$ where $i^*$ is maximal such that $C_{c_j}^{i^*}\neq \emptyset$).
                We add~$C_{c_j}^{\max}$ as well as a shortest~$s$-$c_j$-path~$P$ to~$C_s^i$ (i.e., $C_s^i = C_s^i \cup C_{c_j}^{\max} \cup P$) and reset~$C_{c_j}^{\max} := \emptyset$.
                If~$|C_s^i| > 3k$ afterwards, then add a vertex~$s_i'$ with neighbors~$N(s) \cap C_s^i$ to~$C_s^i$ and increment~$i$.
                
                Afterwards, we process the components from~$\mathcal{C}_s$.
                As long as $\mathcal{C}_s \neq \emptyset$, pick an arbitrary connected component~$C$ from~$\mathcal{C}_s$, delete~$C$ from~$\mathcal{C}_s$, and add $C$ to ${C}_s^i$ (i.e., $C_s^i := C_s^i \cup C$).
                If~$|C_s^i| > k$ or $\mathcal{C}_s = \emptyset$, then add a vertex~$s_i'$ with neighbors~$N(s) \cap C_s^i$ to~$C_s^i$, and increment~$i$.

                \item For each $s \in S$ and $i \in \mathbb{N}$, find a spanning tree~$T$ of $C_s^i$.
                Note that such a spanning tree exists since each~$C_s^i$ is connected due to the vertex~$s_i'$.
                Doubling the edges of~$T$ results in an Eulerian graph~$T'$ which admits an Eulerian walk~$P$.
                For each vertex appearing multiple times in the Eulerian walk, we replace all by its first occurrence with a new vertex connected only to its predecessor and successor in the Eulerian walk, resulting in~$\widetilde C_s^i$.
                The Eulerian walk~$P$ is a Hamiltonian path in~$\widetilde C_s^i$ and we set $H_i := P$.
                Add $\widetilde C_s^i$ to~$G'$ and add all edges between each vertex~$s \in S$ and all of its original neighbors in~$G$.

                \item To achieve that $|S'| = |T_1'| = \dots = |T_\ncc'|$, append paths of appropriate length to $H_i$ or an arbitrary vertex from~$S'$.
            \end{enumerate}

            Note that the connected components of~$G' - S$ are of the form~$\widetilde C_s^i$ for some~$s \in S $ and~$i \in \mathbb{N}$.
            We have $|\widetilde C_s^i| \le 2 |C_s^i| \le 2 (5k+1)$.
            Note that the first inequality holds since the number of copies of each vertex in~$\widetilde C_s^i$ corresponds to its degree in the Eulerian graph~$T'$ (which contains twice as many edges as~$T$) and summing over all degrees in a tree gives twice its number of edges which is one less than its number of vertices. 
            The second inequality holds as $|C_s^i| \le 3k$ whenever a connected component of~$G-S$ is added to~$C_s^i$ plus possible a shortest~$s$-$c_j$-path for a child~$c_j$ of~$s$ (and this path has length at most~$k$ as $c_j$ is a child of~$s$).
            It follows that each connected component of~$G' - S $ has $O(k)$ vertices.

            Next, we analyze the number of connected components of~$G' -S$.
            Because each connected component except for possible the last one contains at least $k$ vertices, it follows that the number of connected components is at most~$n/k + 1 = O(n/k)$.
                      
            It remains to show how to compute the distance matrix~$D$ for~$G$ from the distance matrix~$D'$ from~$G'$.
        Let $V$ be the vertex set of~$G$ and~$V'$ be the vertex set of~$G'$.
        Note that any vertex~$v \in V$ is also contained in~$V'$.
        We claim that~$D [u, v] = D' [u, v]$.
        Clearly, any $u$-$v$-path in~$G$ also exists in~$G'$, so we have $D[u, v] \ge D' [u, v]$.
        For any $u$-$v$-path in~$G'$, we can get a $u$-$v$-walk in~$G$ by replacing vertices from~$V' \setminus V$ with the vertex due to which they were added (i.e., $s_i'$ gets replaced by~$s$, and the vertex added due to the $i$-th occurrence of $v$ in some Eulerian walk $P$ is replaced by~$v$), implying that $D' [u, v] \ge D [u, v]$.
        \end{proof}       
Using \Cref{lem:apsp-nice-partition}, we now show that APSP on unweighted graphs can be solved faster than~$O(n^{\omega})$ if the vertex integrity is sufficiently small and $\omega > 2$.

\apsp*
\begin{proof}
    If $k \ge {n}^{0.6}$, then $k^{(\omega -1)/2} n^2 > n^\omega$ and we apply the $O(n^\omega \log n)$-time algorithm \cite{DBLP:conf/stoc/Seidel92}.
    Otherwise, $k \le {n}^{0.6}$ and
    we use the following algorithm to solve APSP.
    \begin{enumerate}
        \item\label{step:apsp-good-components}
        We apply \Cref{lem:apsp-nice-partition}, resulting in a graph~$G$ and sets~$S, T_1, \dots, T_\ncc$.
        \item\label{step:apsp-components}
        For every $i$, we solve APSP on $G[S \cup T_i]$.
        Let $D_i = D_i'[T_i, S]$, where $D_i'$ is the distance matrix of $G[S \cup T_i]$.
        \item\label{step:apsp-weighted-component}
        We compute an edge-weighted graph~$G_S$ where the vertex set is $S$, the edge set is $\binom{S}{2}$, and~$w(uv) = 1$ if $uv \in E$, and $w(uv) = \min_{i} D_i(uv)$ otherwise.
        We solve APSP on this weighted graph and call the resulting distance matrix $D_S$.
        As we show later, $D_S[u, v]$ is the distance from $u$ to $v$ in $G$.
        \item\label{step:apsp-final}
        For each $i, j \in [\ncc]$, we compute $D_i^* := D_i \star D_{S}$
        and $D_{i,j}^* := D_i^*\star D_j^T$.
        \item \label{step:apsp-return}
        Return a symmetric matrix~$D^*$ where the upper triangular part of $D^*$ is defined by
        \begin{align*} D^* [u, v] = \begin{cases} D_S [u, v] & u, v \in S,\\
        D_i^* [u, v] & u \in S \text{ and } v \in T_i,\\
        \min \{D_i [u, v], D_{i, i}^* [u, v]\} & u, v \in T_i \text{ for some }i \in [\ncc],\\
        D_{i,j}^* [u, v] & u \in T_i \text{ and } v \in T_j \text{ for } i \neq j.
        \end{cases}\end{align*}
    \end{enumerate}
 
    First we show the correctness of the algorithm.
    For $u, v\in S$, Step~\ref{step:apsp-weighted-component} guarantees that the distance between~$u$ and~$v$ in the weighted graph~$G_S$ is~${D^* [u, v] = D_S [u, v]}$.
    As each edge~$u'v'$ in~$G_S$ corresponds to a path in~$S \cup T_i$ of length~$w(u'v')$ for some~$i \in [\ncc]$, it follows that each path in~$G_S$ corresponds to a walk in~$G$ of the same length.
    Thus, $D_S [u, v]$ is not smaller than a shortest $u$-$v$-path in~$G$.
    On the other hand, let $P$ be a shortest $u$-$v$-path in~$G$.
    Let $u = s_1, \dots, s_\ell = v$ be the vertices of~$P$ in~$S$ (in the order they appear in~$P$).
    By construction, for each~$i < \ell$, there are no vertices in $S$ between~$s_i$ and~$s_{i+1}$ in~$P$. Hence, this subpath of~$P$ is fully contained in $G[S \cup T_j]$ for some connected component~$T_j$, and we have~$w(s_i, s_{i+1}) \leq \dist_{G[S \cup T_j]} (s_i, s_{i+1})$.
    Consequently, we have that $P$ corresponds to a path in the weighted graph~$G_S$ and $D_S (u, v)$ is not larger than the length of a shortest $u$-$v$-path in~$G$.
    Thus, $D_S [u, v]$ equals the distance between $u$ and $v$ in~$G$ and~$D^* [u, v]$ has therefore been computed correctly.

    Next, assume that $u \in T_i$ for some~$i \in [\ncc]$ and $v \in S$ (the case $u \in S$ and $v \in T_i$ is analogous).
    A shortest $u$-$v$-path then consists of a shortest~$u$-$s$-path in~$G[ S \cup T_i]$ and a shortest~$s$-$v$-path in~$G$ for some~$s \in S$ (since each vertex in $T_i$ only has neighbors in~$T_i \cup S$).
    This implies~${\dist_G[u, v] = \min_{s \in S} D_i [u, s] + D_S [s, v] = D_i^* [u, v]}$.

    Next, assume that $u \in T_i$ and $v \in T_i$ for some~$i \in [\ncc]$.
    Then a shortest $u$-$v$-path either stays completely in~$T_i$ or it decomposes into a shortest~$u$-$s$-path in~$G[S \cup T_i]$, followed by a shortest~\mbox{$s$-$s'$-path} in~$G$, and finally followed by a shortest~$s'$-$v$-path in~$G[S \cup T_i]$ for some~${s, s' \in S}$.
    In the first case, the distance between $u$ and $v$ equals to~$D_i [u, v]$.
    In the second case, the distance between $u$ and $v$ equals $${\min_{s, s' \in S} \dist_{G[S \cup T_i]} (u, s) + \dist_{G} (s, s') + \dist_{G[S \cup T_i]} (s', v) = D_{i, i}^* [u, v]}.$$
    Thus, $D^* [u, v]$ contains the distance between $u$ and $v$.

    Finally, assume that $u \in T_i$ and $v \in T_j$ for some~$i \neq j \in [\ncc]$.
    A shortest $u$-$v$-path then decomposes into a shortest~$u$-$s$-path in~$G[S \cup T_i]$, followed by a shortest~$s$-$s'$-path in~$G$, and finally followed by a shortest~$s'$-$v$-path in~$G[S \cup T_j]$ for some~$s, s' \in S$.
    Consequently, we have that~${\dist (u, v) = \min_{s, s' \in S} D_i [u, s] + D_S [s, s'] + D_j [s', v] = D_{i, j}^* [u, v]}$.

    It remains to analyze the running time of the algorithm.
    Step~\ref{step:apsp-good-components} runs in linear time by \cref{lem:apsp-nice-partition}.
    Step~\ref{step:apsp-components} solves APSP on $O(n/k)$ many instances, each with~$O(k)$~vertices.
    As APSP on unweighted graphs with $k$ vertices can be solved in $O (k^\omega \log(k))$ time, Step~\ref{step:apsp-components} runs in $O(n k^{\omega -1} \log(k))$ time.

    Step~\ref{step:apsp-weighted-component} first computes a weighted graph on $O(k)$ vertices.
    Each edge can be computed in $O(n/k)$ time.
    As there are $O(k^2)$ many edges, computing the graph takes~$O( n \cdot k)$ time.
    Solving weighted APSP on this graph then can be done in $O(k^3)$ time~\cite{DBLP:journals/siamcomp/Williams18}.
    As we assumed~$k \le n^{0.6}$, this step runs in $O(n^{1.8})$ time.

    Step~\ref{step:apsp-final} computes for each pair~$(i, j) \in [n/k]$ the min-plus product of $D_i$, $D_S$, and $D_j$.
    We will show that we can compute these min-plus products in $O(k^{(3+\omega)/2})$ time (which is faster than the state-of-the-art algorithm for computing arbitrary min-plus products).
    The trick is to use the fact that~$D_i$ and~$D^*_i$ have bounded difference and then use a result by
    Chi et al.~\cite{DBLP:conf/stoc/ChiDX022} stating that the min-plus product of two matrices of dimension $n \times n$ can be computed in~${O(n^{(3+\omega)/2}) \subseteq O(n^{2.687})}$~time if one of the matrices has bounded difference.
    While it is not true a priori that~$D_i$ and~$D^*_i$ have bounded difference, we will order the rows according to the Hamiltonian path~$H_i$.
    This ensures that two consecutive vertices are adjacent, which implies that consecutive entries in a row differ by at most one.
    Hence, each computation of the min-plus product can be done in~$O (k^{(3+\omega)/2})$~time, as~$|S| = |T_1 | = \dots = |T_\ncc| = O(k)$ by \Cref{lem:apsp-nice-partition}.
    Overall, the running time of Step~\ref{step:apsp-final} is in~${O((n/k)^2 \cdot k^{(3+\omega)/2}) = O(n^2 \cdot k^{(\omega-1)/2})}$.

    Lastly, Step~\ref{step:apsp-return} runs in $O (n^2)$ time.
    The overall running time of $O(k^{(\omega -1)/2} \cdot n^2)$ follows from the observation that the running time for Step~\ref{step:apsp-final} dominates the running time of Steps~\ref{step:apsp-good-components},~\ref{step:apsp-components},~\ref{step:apsp-weighted-component}, and~\ref{step:apsp-return} as~$k \le n$ and~$\omega < 2.9$.
\end{proof}

To conclude this section, we give an adaptive algorithm for constant diameter.
The crucial observation here is that we can compute the product of the adjacency matrix of a graph with any other matrix in $O(\iota^{\omega -2} n^2)$ time.

\begin{lemma}\label{lem:multiplication-with-arbitrary-matrix}
 Given a graph $G$ with adjacency matrix~$A$, a $k$-separator $S$ for~$G$, and an~$n \times n$-matrix~$M$, the matrices~$A M$ and $MA$ can be computed in $O(k^{\omega -2} n^2)$ time.
\end{lemma}
\begin{proof}
Suppose that the adjacency matrix has the form as described in \Cref{eq:adj-matrix}.
Then,
\begin{align*}
    AM = \begin{bmatrix}
        \gamma M[S, V] + \beta M[\overline{S}, V] \\
        \beta^T M[S, V] + \alpha M[\overline{S}, V]
    \end{bmatrix}.
\end{align*}

A closer inspection reveals that all the computation can be done in $O(k^{\omega - 2} n^2)$ time:
First observe that $\gamma M[S, V]$ can be computed in $O (k^{\omega - 1} n)$ time.
For the other terms, note that
\begin{align*}
    &\beta M[\overline{S}, V] = \begin{bmatrix}
        \beta M[\overline{S}, S] & \beta M[\overline{S}, T_1] & \ldots & \beta M[\overline{S}, T_{\nu}]
    \end{bmatrix}, \\
    &\beta^T M[S, V] = \begin{bmatrix}
        \beta^T M[\overline{S}, S] & \beta^T M[\overline{S}, T_1] & \ldots & \beta^T M[\overline{S}, T_{\nu}]
    \end{bmatrix}, \\
    &\alpha M[\overline{S}, V] = \begin{bmatrix}
        \alpha_1 M[T_1, \overline{S}] &  \alpha_2 M[T_2, \overline{S}] & \ldots &  \alpha_{\nu} M[T_{\ncc}, \overline{S}]
    \end{bmatrix}^T.
\end{align*}
Since there are $O(n/k)$ submatrices and each takes~$O(k^{\omega - 1} n)$ time to compute, $AM$ can be computed in $O(k^{\omega - 2} n^2)$ time.
Note that $MA = (AM^T)^T$ can be computed analogously.
\end{proof}

We obtain our algorithm from the folklore observation that for any two vertices~$u, v \in V$, there is a walk of length exactly length $d$ between $u$ and $v$ if and only if $A^d[u, v] \ne \emptyset$.

\begin{proposition}
    Given a graph $G$ and a~$k$-separator~$S$, APSP can be solved in $O(d k^{\omega - 2} n^2)$ time where $d$ is the diameter of~$G$.
\end{proposition}

\begin{proof}
 Let~$A$ be the adjacency matrix of~$G$.
 We compute matrices~$B^1, \dots, B^d\in \{0,1 \}^{n \times n}$ recursively as follows.
 We start with~$B^1 := A$.
 Matrix~$B^{i+1}$ is computed by multiplying~$B^i$ with~$A$ and replacing all non-zero entries by~$1$.
 Note that $A^i[u, v] = 0$ if and only if~$B^i [u, v] = 0$.
 Thus, the distance between two vertices~$u \neq v$ is the minimum $i$ such that~$B^i[u, v] \neq 0$.
 By \Cref{lem:multiplication-with-arbitrary-matrix}, we can multiply any matrix with~$A$ in $O(k^{\omega -2} n^2)$ time.
 Thus, we can compute~$B^1,\dots, B^d$ in~$O(d k^{\omega -2} n^2)$~time.
\end{proof}
Note that $d$ can be of order $\Omega(k^2)$ as the vertex integrity in a cycle with~$n$ vertices is in~$O(\sqrt{n})$ and the diameter of the cycle is~$\lfloor n / 2 \rfloor$.

\section{Conclusion}
\label{sec:conclusion}

In this work, we investigated the parameter vertex integrity in search for more efficient algorithms in the FPT-in-P paradigm.
We exhibited that for many problems, the structure of graphs with a small vertex integrity allows us to harness the power of fast matrix multiplication.
In particular, we designed an~$O(\iota^{\omega - 1}n)$-time algorithm for computing the girth of a graph and randomized $O(\iota^{\omega - 1}n)$-time algorithms for finding a four-vertex subgraph (which is not a $K_4$ or a $\overline{K_4}$) and a maximum matching.
We also showed that unweighted APSP can be solved in $O(\iota^{\nicefrac{\omega - 1}{2}} n^2)$ time using min-plus product of bounded-differences matrices,
leaving open whether there is an $O(\iota^{\omega - 2} n^2)$-time algorithm.
More broadly, we wonder whether a similar approach using fast matrix multiplication can be used for graphs of bounded tree-depth or tree-width.
Existing approaches (e.g.~\cite{DBLP:journals/siamcomp/ChibaN85,DBLP:journals/talg/FominLSPW18,DBLP:conf/stacs/IwataOO18}) do not seem amenable to fast matrix multiplication.

\bibliography{ref}

\begin{thebibliography}{45}
\providecommand{\natexlab}[1]{#1}
\providecommand{\url}[1]{\texttt{#1}}
\expandafter\ifx\csname urlstyle\endcsname\relax
  \providecommand{\doi}[1]{doi: #1}\else
  \providecommand{\doi}{doi: \begingroup \urlstyle{rm}\Url}\fi

\bibitem[Abboud et~al.(2016)Abboud, Williams, and
  Wang]{DBLP:conf/soda/AbboudWW16}
Amir Abboud, Virginia~Vassilevska Williams, and Joshua~R. Wang.
\newblock Approximation and fixed parameter subquadratic algorithms for radius
  and diameter in sparse graphs.
\newblock In \emph{Proceedings of the 27th Annual {ACM-SIAM} Symposium on
  Discrete Algorithms ({SODA}~'16)}, pages 377--391. {SIAM}, 2016.
\newblock \doi{10.1137/1.9781611974331.ch28}.
\newblock URL \url{https://doi.org/10.1137/1.9781611974331.ch28}.

\bibitem[Alon and Yuster(2007)]{DBLP:conf/esa/AlonY07}
Noga Alon and Raphael Yuster.
\newblock Fast algorithms for maximum subset matching and all-pairs shortest
  paths in graphs with a (not so) small vertex cover.
\newblock In \emph{Proceedings of the 15th Annual European Symposium on
  Algorithms ({ESA}~'07)}, pages 175--186, 2007.
\newblock \doi{10.1007/978-3-540-75520-3\_17}.
\newblock URL \url{https://doi.org/10.1007/978-3-540-75520-3\_17}.

\bibitem[Alon et~al.(1995)Alon, Yuster, and Zwick]{DBLP:journals/jacm/AlonYZ95}
Noga Alon, Raphael Yuster, and Uri Zwick.
\newblock Color-coding.
\newblock \emph{J. {ACM}}, 42\penalty0 (4):\penalty0 844--856, 1995.
\newblock \doi{10.1145/210332.210337}.
\newblock URL \url{https://doi.org/10.1145/210332.210337}.

\bibitem[Bar{-}Yehuda et~al.(1998)Bar{-}Yehuda, Geiger, Naor, and
  Roth]{DBLP:journals/siamcomp/Bar-YehudaGNR98}
Reuven Bar{-}Yehuda, Dan Geiger, Joseph Naor, and Ron~M. Roth.
\newblock Approximation algorithms for the feedback vertex set problem with
  applications to constraint satisfaction and bayesian inference.
\newblock \emph{{SIAM} J. Comput.}, 27\penalty0 (4):\penalty0 942--959, 1998.
\newblock \doi{10.1137/S0097539796305109}.
\newblock URL \url{https://doi.org/10.1137/S0097539796305109}.

\bibitem[Barefoot et~al.(1987)Barefoot, Entringer, and
  Swart]{barefoot1987vulnerability}
Curtis~A Barefoot, Roger Entringer, and Henda Swart.
\newblock Vulnerability in graphs - {A} comparative survey.
\newblock \emph{Journal of Combinatorial Mathematics and Combinatorial
  Computing}, 1\penalty0 (38):\penalty0 13--22, 1987.

\bibitem[Benko et~al.(2009)Benko, Ernst, and
  Lanphier]{DBLP:journals/siamdm/BenkoEL09}
D.~Benko, C.~Ernst, and Dominic Lanphier.
\newblock Asymptotic bounds on the integrity of graphs and separator theorems
  for graphs.
\newblock \emph{{SIAM} Journal on Discrete Mathematics}, 23\penalty0
  (1):\penalty0 265--277, 2009.
\newblock \doi{10.1137/070692698}.
\newblock URL \url{https://doi.org/10.1137/070692698}.

\bibitem[Bentert et~al.(2019)Bentert, Fluschnik, Nichterlein, and
  Niedermeier]{DBLP:journals/jcss/BentertFNN19}
Matthias Bentert, Till Fluschnik, Andr{\'{e}} Nichterlein, and Rolf
  Niedermeier.
\newblock Parameterized aspects of triangle enumeration.
\newblock \emph{Journal of Computer and System Sciences}, 103:\penalty0 61--77,
  2019.
\newblock \doi{10.1016/j.jcss.2019.02.004}.
\newblock URL \url{https://doi.org/10.1016/j.jcss.2019.02.004}.

\bibitem[Bodlaender et~al.(2020)Bodlaender, Hanaka, Kobayashi, Kobayashi,
  Okamoto, Otachi, and van~der
  Zanden]{DBLP:journals/algorithmica/BodlaenderHKKOO20}
Hans~L. Bodlaender, Tesshu Hanaka, Yasuaki Kobayashi, Yusuke Kobayashi, Yoshio
  Okamoto, Yota Otachi, and Tom~C. van~der Zanden.
\newblock Subgraph isomorphism on graph classes that exclude a substructure.
\newblock \emph{Algorithmica}, 82\penalty0 (12):\penalty0 3566--3587, 2020.
\newblock \doi{10.1007/s00453-020-00737-z}.
\newblock URL \url{https://doi.org/10.1007/s00453-020-00737-z}.

\bibitem[Bringmann et~al.(2020)Bringmann, Husfeldt, and
  Magnusson]{DBLP:journals/algorithmica/BringmannHM20}
Karl Bringmann, Thore Husfeldt, and M{\aa}ns Magnusson.
\newblock Multivariate analysis of orthogonal range searching and graph
  distances.
\newblock \emph{Algorithmica}, 82\penalty0 (8):\penalty0 2292--2315, 2020.
\newblock \doi{10.1007/s00453-020-00680-z}.
\newblock URL \url{https://doi.org/10.1007/s00453-020-00680-z}.

\bibitem[Bunch and Hopcroft(1974)]{bunch1974triangular}
James~R. Bunch and John~E. Hopcroft.
\newblock Triangular factorization and inversion by fast matrix multiplication.
\newblock \emph{Mathematics of Computation}, 28\penalty0 (125):\penalty0
  231--236, 1974.
\newblock \doi{10.1090/S0025-5718-1974-0331751-8}.

\bibitem[B{\"{u}}rgisser et~al.(1997)B{\"{u}}rgisser, Clausen, and
  Shokrollahi]{DBLP:books/daglib/0090316}
Peter B{\"{u}}rgisser, Michael Clausen, and Mohammad~Amin Shokrollahi.
\newblock \emph{Algebraic complexity theory}.
\newblock Springer, 1997.
\newblock ISBN 3-540-60582-7.
\newblock \doi{0.1007/978-3-662-03338-8}.

\bibitem[Chi et~al.(2022)Chi, Duan, Xie, and Zhang]{DBLP:conf/stoc/ChiDX022}
Shucheng Chi, Ran Duan, Tianle Xie, and Tianyi Zhang.
\newblock Faster min-plus product for monotone instances.
\newblock In \emph{Proceedings of the 54th Annual {ACM} {SIGACT} Symposium on
  Theory of Computing ({STOC}~'22)}, pages 1529--1542. {ACM}, 2022.
\newblock \doi{10.1145/3519935.3520057}.
\newblock URL \url{https://doi.org/10.1145/3519935.3520057}.

\bibitem[Chiba and Nishizeki(1985)]{DBLP:journals/siamcomp/ChibaN85}
Norishige Chiba and Takao Nishizeki.
\newblock Arboricity and subgraph listing algorithms.
\newblock \emph{SIAM Journal on Computing}, 14\penalty0 (1):\penalty0 210--223,
  1985.
\newblock \doi{10.1137/0214017}.
\newblock URL \url{https://doi.org/10.1137/0214017}.

\bibitem[Corneil et~al.(1985)Corneil, Perl, and
  Stewart]{DBLP:journals/siamcomp/CorneilPS85}
Derek~G. Corneil, Yehoshua Perl, and Lorna~K. Stewart.
\newblock A linear recognition algorithm for cographs.
\newblock \emph{SIAM Journal on Computing}, 14\penalty0 (4):\penalty0 926--934,
  1985.
\newblock \doi{10.1137/0214065}.
\newblock URL \url{https://doi.org/10.1137/0214065}.

\bibitem[Coudert et~al.(2019)Coudert, Ducoffe, and
  Popa]{DBLP:journals/talg/CoudertDP19}
David Coudert, Guillaume Ducoffe, and Alexandru Popa.
\newblock Fully polynomial {FPT} algorithms for some classes of bounded
  clique-width graphs.
\newblock \emph{{ACM} Transactions on Algorithms}, 15\penalty0 (3):\penalty0
  33:1--33:57, 2019.
\newblock \doi{10.1145/3310228}.
\newblock URL \url{https://doi.org/10.1145/3310228}.

\bibitem[Deng et~al.(2022)Deng, Kirkpatrick, Rong, Williams, and
  Zhong]{DBLP:conf/icalp/DengKRWZ22}
Mingyang Deng, Yael Kirkpatrick, Victor Rong, Virginia~Vassilevska Williams,
  and Ziqian Zhong.
\newblock New additive approximations for shortest paths and cycles.
\newblock In \emph{Proceedings of the 49th International Colloquium on
  Automata, Languages, and Programming ({ICALP~'22})}, pages 50:1--50:10.
  Schloss Dagstuhl - Leibniz-Zentrum f{\"{u}}r Informatik, 2022.
\newblock \doi{10.4230/LIPIcs.ICALP.2022.50}.
\newblock URL \url{https://doi.org/10.4230/LIPIcs.ICALP.2022.50}.

\bibitem[Drange et~al.(2016)Drange, Dregi, and van~'t
  Hof]{DBLP:journals/algorithmica/DrangeDH16}
P{\aa}l~Gr{\o}n{\aa}s Drange, Markus~S. Dregi, and Pim van~'t Hof.
\newblock On the computational complexity of vertex integrity and component
  order connectivity.
\newblock \emph{Algorithmica}, 76\penalty0 (4):\penalty0 1181--1202, 2016.
\newblock \doi{10.1007/s00453-016-0127-x}.
\newblock URL \url{https://doi.org/10.1007/s00453-016-0127-x}.

\bibitem[Dvor{\'{a}}k et~al.(2021)Dvor{\'{a}}k, Eiben, Ganian, Knop, and
  Ordyniak]{DBLP:journals/ai/DvorakEGKO21}
Pavel Dvor{\'{a}}k, Eduard Eiben, Robert Ganian, Dusan Knop, and Sebastian
  Ordyniak.
\newblock The complexity landscape of decompositional parameters for {ILP}:
  {P}rograms with few global variables and constraints.
\newblock \emph{Artificial Intelligence}, 300:\penalty0 103561, 2021.
\newblock \doi{10.1016/j.artint.2021.103561}.
\newblock URL \url{https://doi.org/10.1016/j.artint.2021.103561}.

\bibitem[Eirinakis et~al.(2017)Eirinakis, Williamson, and
  Subramani]{DBLP:journals/jgaa/EirinakisWS17}
Pavlos Eirinakis, Matthew~D. Williamson, and K.~Subramani.
\newblock On the {S}hoshan-{Z}wick algorithm for the all-pairs shortest path
  problem.
\newblock \emph{Journal of Graph Algorithms and Applications}, 21\penalty0
  (2):\penalty0 177--181, 2017.
\newblock \doi{10.7155/jgaa.00410}.
\newblock URL \url{https://doi.org/10.7155/jgaa.00410}.

\bibitem[Eisenbrand and Grandoni(2004)]{DBLP:journals/tcs/EisenbrandG04}
Friedrich Eisenbrand and Fabrizio Grandoni.
\newblock On the complexity of fixed parameter clique and dominating set.
\newblock \emph{Theoretical Computer Science}, 326\penalty0 (1-3):\penalty0
  57--67, 2004.
\newblock \doi{10.1016/j.tcs.2004.05.009}.
\newblock URL \url{https://doi.org/10.1016/j.tcs.2004.05.009}.

\bibitem[Fischer and Meyer(1971)]{DBLP:conf/focs/FischerM71}
Michael~J. Fischer and Albert~R. Meyer.
\newblock Boolean matrix multiplication and transitive closure.
\newblock In \emph{Proceedings of the 12th Annual Symposium on Switching and
  Automata Theory ({SWAT~'71})}, pages 129--131. {IEEE} Computer Society, 1971.
\newblock \doi{10.1109/SWAT.1971.4}.
\newblock URL \url{https://doi.org/10.1109/SWAT.1971.4}.

\bibitem[Fomin et~al.(2018)Fomin, Lokshtanov, Saurabh, Pilipczuk, and
  Wrochna]{DBLP:journals/talg/FominLSPW18}
Fedor~V. Fomin, Daniel Lokshtanov, Saket Saurabh, Michal Pilipczuk, and Marcin
  Wrochna.
\newblock Fully polynomial-time parameterized computations for graphs and
  matrices of low treewidth.
\newblock \emph{{ACM} Transactions on Algorithms}, 14\penalty0 (3):\penalty0
  34:1--34:45, 2018.
\newblock \doi{10.1145/3186898}.
\newblock URL \url{https://doi.org/10.1145/3186898}.

\bibitem[Gall(2012)]{DBLP:conf/focs/Gall12}
Fran{\c{c}}ois~Le Gall.
\newblock Faster algorithms for rectangular matrix multiplication.
\newblock In \emph{Proceedings of the 53rd Annual {IEEE} Symposium on
  Foundations of Computer Science ({FOCS~'12})}, pages 514--523. {IEEE}
  Computer Society, 2012.
\newblock \doi{10.1109/FOCS.2012.80}.
\newblock URL \url{https://doi.org/10.1109/FOCS.2012.80}.

\bibitem[Gima et~al.(2022)Gima, Hanaka, Kiyomi, Kobayashi, and
  Otachi]{DBLP:journals/tcs/GimaHKKO22}
Tatsuya Gima, Tesshu Hanaka, Masashi Kiyomi, Yasuaki Kobayashi, and Yota
  Otachi.
\newblock Exploring the gap between treedepth and vertex cover through vertex
  integrity.
\newblock \emph{Theor. Comput. Sci.}, 918:\penalty0 60--76, 2022.
\newblock \doi{10.1016/j.tcs.2022.03.021}.
\newblock URL \url{https://doi.org/10.1016/j.tcs.2022.03.021}.

\bibitem[Godsil(1993)]{DBLP:books/daglib/0077284}
Chris~D. Godsil.
\newblock \emph{Algebraic combinatorics}.
\newblock Chapman and Hall, 1993.
\newblock ISBN 978-0-412-04131-0.

\bibitem[Harvey(2009)]{DBLP:journals/siamcomp/Harvey09}
Nicholas J.~A. Harvey.
\newblock Algebraic algorithms for matching and matroid problems.
\newblock \emph{{SIAM} Journal on Computing}, 39\penalty0 (2):\penalty0
  679--702, 2009.
\newblock \doi{10.1137/070684008}.
\newblock URL \url{https://doi.org/10.1137/070684008}.

\bibitem[Hegerfeld and Kratsch(2019)]{DBLP:conf/icalp/HegerfeldK19}
Falko Hegerfeld and Stefan Kratsch.
\newblock On adaptive algorithms for maximum matching.
\newblock In \emph{Proceedings of the 46th International Colloquium on
  Automata, Languages, and Programming ({ICALP}~'19)}, pages 71:1--71:16.
  Schloss Dagstuhl - Leibniz-Zentrum f{\"{u}}r Informatik, 2019.
\newblock \doi{10.4230/LIPIcs.ICALP.2019.71}.
\newblock URL \url{https://doi.org/10.4230/LIPIcs.ICALP.2019.71}.

\bibitem[Itai and Rodeh(1978)]{DBLP:journals/siamcomp/ItaiR78}
Alon Itai and Michael Rodeh.
\newblock Finding a minimum circuit in a graph.
\newblock \emph{{SIAM} J. Comput.}, 7\penalty0 (4):\penalty0 413--423, 1978.
\newblock \doi{10.1137/0207033}.
\newblock URL \url{https://doi.org/10.1137/0207033}.

\bibitem[Iwata et~al.(2018)Iwata, Ogasawara, and
  Ohsaka]{DBLP:conf/stacs/IwataOO18}
Yoichi Iwata, Tomoaki Ogasawara, and Naoto Ohsaka.
\newblock On the power of tree-depth for fully polynomial {FPT} algorithms.
\newblock In \emph{Proceedings of the 35th Symposium on Theoretical Aspects of
  Computer Science ({STACS}~'18)}, pages 41:1--41:14. Schloss Dagstuhl -
  Leibniz-Zentrum f{\"{u}}r Informatik, 2018.
\newblock \doi{10.4230/LIPIcs.STACS.2018.41}.
\newblock URL \url{https://doi.org/10.4230/LIPIcs.STACS.2018.41}.

\bibitem[Kratsch and Nelles(2020)]{DBLP:conf/stacs/KratschN20}
Stefan Kratsch and Florian Nelles.
\newblock Efficient parameterized algorithms for computing all-pairs shortest
  paths.
\newblock In \emph{Proceedings of the 37th International Symposium on
  Theoretical Aspects of Computer Science ({STACS}~'20)}, pages 38:1--38:15.
  Schloss Dagstuhl - Leibniz-Zentrum f{\"{u}}r Informatik, 2020.
\newblock \doi{10.4230/LIPIcs.STACS.2020.38}.
\newblock URL \url{https://doi.org/10.4230/LIPIcs.STACS.2020.38}.

\bibitem[Lampis and Mitsou(2021)]{DBLP:conf/isaac/LampisM21}
Michael Lampis and Valia Mitsou.
\newblock Fine-grained meta-theorems for vertex integrity.
\newblock In \emph{Proceedings of the 32nd International Symposium on
  Algorithms and Computation ({ISAAC}~'21)}, pages 34:1--34:15, 2021.
\newblock \doi{10.4230/LIPIcs.ISAAC.2021.34}.
\newblock URL \url{https://doi.org/10.4230/LIPIcs.ISAAC.2021.34}.

\bibitem[Lee(2019)]{DBLP:journals/mp/Lee19}
Euiwoong Lee.
\newblock Partitioning a graph into small pieces with applications to path
  transversal.
\newblock \emph{Mathematical Programming}, 177\penalty0 (1-2):\penalty0 1--19,
  2019.
\newblock \doi{10.1007/s10107-018-1255-7}.
\newblock URL \url{https://doi.org/10.1007/s10107-018-1255-7}.

\bibitem[Lov{\'{a}}sz(1979)]{DBLP:conf/fct/Lovasz79}
L{\'{a}}szl{\'{o}} Lov{\'{a}}sz.
\newblock On determinants, matchings, and random algorithms.
\newblock In \emph{Proceedings of the 2nd International Symposium on
  Fundamentals of Computation Theory ({FCT~'79})}, pages 565--574.
  Akademie-Verlag, Berlin, 1979.

\bibitem[Monien(1983)]{DBLP:journals/computing/Monien83}
Burkhard Monien.
\newblock The complexity of determining a shortest cycle of even length.
\newblock \emph{Computing}, 31\penalty0 (4):\penalty0 355--369, 1983.
\newblock \doi{10.1007/BF02251238}.
\newblock URL \url{https://doi.org/10.1007/BF02251238}.

\bibitem[Mucha and Sankowski(2004)]{DBLP:conf/focs/MuchaS04}
Marcin Mucha and Piotr Sankowski.
\newblock Maximum matchings via gaussian elimination.
\newblock In \emph{Proceedings of the 45th Symposium on Foundations of Computer
  Science {(FOCS}~'04)}, pages 248--255, 2004.
\newblock \doi{10.1109/FOCS.2004.40}.
\newblock URL \url{https://doi.org/10.1109/FOCS.2004.40}.

\bibitem[Murota(1999)]{murota1999matrices}
Kazuo Murota.
\newblock \emph{Matrices and matroids for systems analysis}.
\newblock Springer Science \& Business Media, 1999.
\newblock \doi{10.1007/978-3-642-03994-2}.

\bibitem[Rabin and Vazirani(1989)]{DBLP:journals/jal/RabinV89}
Michael~O. Rabin and Vijay~V. Vazirani.
\newblock Maximum matchings in general graphs through randomization.
\newblock \emph{J. Algorithms}, 10\penalty0 (4):\penalty0 557--567, 1989.
\newblock \doi{10.1016/0196-6774(89)90005-9}.
\newblock URL \url{https://doi.org/10.1016/0196-6774(89)90005-9}.

\bibitem[Schwartz(1980)]{DBLP:journals/jacm/Schwartz80}
Jacob~T. Schwartz.
\newblock Fast probabilistic algorithms for verification of polynomial
  identities.
\newblock \emph{Journal of the {ACM}}, 27\penalty0 (4):\penalty0 701--717,
  1980.
\newblock \doi{10.1145/322217.322225}.
\newblock URL \url{https://doi.org/10.1145/322217.322225}.

\bibitem[Seidel(1992)]{DBLP:conf/stoc/Seidel92}
Raimund Seidel.
\newblock On the all-pairs-shortest-path problem.
\newblock In \emph{Proceedings of the 24th Annual {ACM} Symposium on Theory of
  Computing ({STOC~'92})}, pages 745--749. {ACM}, 1992.
\newblock \doi{10.1145/129712.129784}.
\newblock URL \url{https://doi.org/10.1145/129712.129784}.

\bibitem[Shoshan and Zwick(1999)]{DBLP:conf/focs/ShoshanZ99}
Avi Shoshan and Uri Zwick.
\newblock All pairs shortest paths in undirected graphs with integer weights.
\newblock In \emph{Proceedings of the 40th Annual Symposium on Foundations of
  Computer Science ({FOCS}~'99)}, pages 605--615, 1999.
\newblock \doi{10.1109/SFFCS.1999.814635}.
\newblock URL \url{https://doi.org/10.1109/SFFCS.1999.814635}.

\bibitem[Tutte(1947)]{tutte1947factorization}
William~T. Tutte.
\newblock The factorization of linear graphs.
\newblock \emph{Journal of the London Mathematical Society}, 1\penalty0
  (2):\penalty0 107--111, 1947.
\newblock \doi{10.1112/jlms/s1-22.2.107}.

\bibitem[Williams(2018)]{DBLP:journals/siamcomp/Williams18}
Richard~Ryan Williams.
\newblock Faster all-pairs shortest paths via circuit complexity.
\newblock \emph{{SIAM} Journal on Computing}, 47\penalty0 (5):\penalty0
  1965--1985, 2018.
\newblock \doi{10.1137/15M1024524}.
\newblock URL \url{https://doi.org/10.1137/15M1024524}.

\bibitem[Williams et~al.(2015)Williams, Wang, Williams, and
  Yu]{DBLP:conf/soda/WilliamsWWY15}
Virginia~Vassilevska Williams, Joshua~R. Wang, Richard~Ryan Williams, and
  Huacheng Yu.
\newblock Finding four-node subgraphs in triangle time.
\newblock In \emph{Proceedings of the 26th Annual {ACM-SIAM} Symposium on
  Discrete Algorithms ({SODA}~'15)}, pages 1671--1680, 2015.
\newblock \doi{10.1137/1.9781611973730.111}.
\newblock URL \url{https://doi.org/10.1137/1.9781611973730.111}.

\bibitem[Yuster and Zwick(1997)]{DBLP:journals/siamdm/YusterZ97}
Raphael Yuster and Uri Zwick.
\newblock Finding even cycles even faster.
\newblock \emph{{SIAM} J. Discret. Math.}, 10\penalty0 (2):\penalty0 209--222,
  1997.
\newblock \doi{10.1137/S0895480194274133}.
\newblock URL \url{https://doi.org/10.1137/S0895480194274133}.

\bibitem[Zippel(1979)]{zippel1979probabilistic}
Richard Zippel.
\newblock Probabilistic algorithms for sparse polynomials.
\newblock In \emph{Proceedings of the 2nd International Symposium on Symbolic
  and Algebraic Manipulation ({EUROSM~'79})}, pages 216--226. Springer, 1979.
\newblock \doi{10.1007/3-540-09519-5\_73}.

\end{thebibliography}

\end{document}